\documentclass[10pt]{article}
\usepackage{color, amssymb, amsthm, amsmath, amsfonts, ascmac, comment, enumerate,cite}

\usepackage{graphicx}
\usepackage{hyperref}
\usepackage{cleveref}

\usepackage{multirow}
\usepackage{tcolorbox}
\usepackage{xcolor}
\hypersetup{
  colorlinks, 
  citecolor=blue!20!black!30!green,
  linkcolor=red,
  urlcolor=blue}
\newtheorem{Thm}{Theorem}
\newtheorem*{Thm*}{Theorem}
\newtheorem*{Thm1}{\rm\bf Theorem~\ref{Thm_asymptotics}}
\newtheorem{Prop}{Proposition}	
\newtheorem{Lem}{Lemma}
\newtheorem{Cor}{Corollary}
\newtheorem{Fact}{Fact}
\newtheorem*{Fact1}{\rm\bf Fact~\ref{Fact_f_dist}}
\newtheorem*{Fact2}{\rm\bf Fact~\ref{Fact_rand_sep}}
\newtheorem{Def}{Definition}
\newtheorem{Rem}{Remark}
\newtheorem*{Rem*}{Remark}
\usepackage[top=30truemm,bottom=30truemm,left=20truemm,right=20truemm]{geometry}

\crefname{Prop}{Proposition}{Propositions}
\crefname{Lem}{Lemma}{Lemmas}
\crefname{Fact}{Fact}{Facts}
\crefname{figure}{Fig}{Figures}
\crefname{Thm}{Theorem}{Theorems}
\crefname{Cor}{Corollary}{Corollaries}
\crefname{Def}{Definition}{Definitions}
\crefname{Algo}{Algorithm}{Algorithms}
\crefname{section}{Section}{Sections}

\newcommand{\Bset}{\{0, 1\}}

\newcommand{\ED}{\overline{D}^\mu}

\newcommand{\ER}{\overline{R}}
\newcommand{\E}{\mathbf{E}}
\newcommand{\calA}{\mathcal{A}}
\newcommand{\calB}{\mathcal{B}}
\newcommand{\Cost}{\mathrm{Cost}}
\newcommand{\Mand}{~~\text{and}~~}

\title{Zero-error expectation equals amortized query complexity}
\author{Daiki Suruga \\IQC, University of Waterloo}

\begin{document}
\maketitle
%\tableofcontents
\begin{abstract}
This paper investigates the direct sum question for expected randomized and distributional query complexity.
Our main result gives an exact characterization of the amortized expected randomized query complexity.
For any total relation $f$ and any error tolerance $\varepsilon \in [0,1]$, we prove
\[
\lim_{n \to \infty} \frac{\overline{R}_\varepsilon(f^n)}{n} = (1 - \varepsilon) \overline{R}_0(f).
\]
Thus the amortization converts bounded-error into zero error with the exact multiplicative factor $1-\varepsilon$.
We also prove corresponding liminf/limsup bounds for worst-case randomized and distributional query complexity.
These results improve prior direct-sum bounds that were known only up to constant factors or in restricted error regimes, and they resolve an open question posed by Blais and Brody (2019).
Additionally for one-sided computation of the function $\operatorname{OR}_n \circ f$, we obtain analogous exact amortized identities for both expected and worst-case cost.
\par
As applications, we obtain separations between amortized and single-instance costs, including unbounded separations for distributional complexity and randomized relations, and a quadratic barrier for randomized total functions.

\end{abstract}

\section{Introduction}
The direct-sum problem asks how the complexity of solving $n$ instances of a problem compares with $n$ times its single-instance complexity, a central question in theoretical computer science. It formally addresses the intuition of resource scaling: does the task of solving $n$ independent instances of a problem $f$ necessarily incur $n$ times the cost of solving a single instance? A positive answer---establishing that $\text{Cost}(f^n) \approx n \cdot \text{Cost}(f)$---confirms that resources cannot be shared across independent inputs, whereas a negative answer implies that resources can be optimized when solving multiple instances simultaneously.
Due to its foundational nature, this question has attracted significant attention across diverse research areas, including communication complexity~\cite{FKNN95, CSWY01, JRS03d, BJKS04, BBCR09, JK09, BR11, JPY12, MWY13, KMSY14, Bra17, Jain20, JK22}, query complexity~\cite{JKS10, ACLT10, Dru11, Mon13, MS15, BK18, BB19, BB20, GM21, BTLS23, BKST24}, and circuit complexity~\cite{IW97, IJKW08, GNW11}.

\par
This paper investigates the direct sum question within the frameworks of classical distributional and randomized query complexity.
We also study a related amortization problem for compositions of the form $\operatorname{OR}_n\circ f$ under one-sided error.
In the setting of standard worst-case randomized query complexity, several standard bounds are known.
These include the upper bound $R_\varepsilon(f^n) \leq n \cdot R_{\varepsilon/n}(f)$, where $R_\varepsilon(f)$ denotes the worst-case randomized query complexity of a relation $f$ with error at most $\varepsilon$, as well as lower bounds of the form $R_\varepsilon(f^n) = \Omega(n \cdot R_\varepsilon(f))$.
Ref.~\cite{JKS10} further contributes to the stronger inequality:
\begin{equation}\label{eq_intro_1}
 \delta^2 n \cdot R_{\frac{\varepsilon}{1-\delta} + \delta}(f) \leq R_\varepsilon(f^n).
\end{equation}
On the other hand, recent work has explored expected randomized query complexity, where the cost is averaged over the algorithm's internal randomness.
In this context, Ref.~\cite{BK18} improves the direct sum bound~\eqref{eq_intro_1} to $n \cdot \ER_\varepsilon(f^n) \le \ER_\varepsilon(f^n)$, where $\ER_\varepsilon(f)$ is the expected randomized query complexity for a relation $f$ with error at most $\varepsilon$.
Furthermore, Ref.~\cite{BB19} establishes an optimal direct sum theorem for expected complexity, showing that $\ER_\varepsilon(f^n) = \Theta(n \ER_{\varepsilon/n}(f))$ holds for $\varepsilon \le 1/20$.
This was recently generalized to an optimal direct product theorem, which imposes a stronger requirement than direct sum theorems, by Ref.~\cite{BB25}.

\par
Analogous results exist in the setting of distributional query complexity, where inputs are drawn according to a distribution $\mu$ and the distributional query complexity is denoted by $D^\mu_\varepsilon(f)$.
Standard folklore bounds include the lower bound $D^{\mu^n}_{\varepsilon}(f^n) = \Omega(n \cdot D_{\varepsilon}^\mu(f))$ and the upper bound $D^{\mu^n}_{\varepsilon}(f^n) = O(n \cdot D^\mu_{\varepsilon/n}(f))$.
A less trivial result is established in Ref.~\cite{Sha01}, which demonstrates the existence of a function $f$ satisfying\footnote{Some papers explain the result~\cite{Sha01} as
$D^{\mu^n}_{\varepsilon}(f^n) = O\left(\varepsilon D_{\frac{\varepsilon}{n}}^\mu(f) \right)$, which is misleading and not correct when $n$ is not fixed.}
\begin{equation*}
D^{\mu^n}_{\varepsilon}(f^n) = O\left(\varepsilon D_{\frac{\varepsilon}{n}}^\mu(f) + n \log \frac{n}{\varepsilon}\right).
\end{equation*}
Consequently, a general direct-\emph{product} theorem fails for distributional query complexity,
 although Ref.~\cite{Dru11} has shown that a specific variant of the direct product theorem does hold.
In the context of expected distributional query complexity, where the expected distributional query complexity is denoted by $\ED_\varepsilon(f)$, Ref.~\cite{BKST24} recently established the following direct sum theorem:
\begin{equation}\label{eq_intro_2}
\overline{D}^{\mu^n}_{\varepsilon}(f^n) = \tilde{\Omega}(\varepsilon^2 n) \overline{D}^{\mu}_{\Theta(\varepsilon/n)}(f),
\end{equation}
which nearly matches the folklore upper bound $\overline{D}^{\mu^n}_{\varepsilon}(f^n) = O(n \overline{D}^\mu_{\varepsilon/n}(f))$ up to constant factors when $\varepsilon$ is fixed.
\par

\par
The preceding results characterize the direct-sum behavior up to constant or polylogarithmic factors, but they do not identify the exact asymptotic coefficient in the global-error expected-cost model.  In particular, the exact value of
\[
\lim_{n\to\infty}
\frac{\ER_\varepsilon(f^n)}{n}
\]
was not known.  Ref.~\cite{BB19} also left open whether their strong direct-sum theorem for expected query complexity has an analogue in the worst-case randomized query model.

Our first main theorem identifies the exact expected-cost limit as
\[
(1-\varepsilon)\ER_0(f)
\]
for every total relation and every \(\varepsilon\in[0,1]\).  Its worst-case upper bound also gives a negative answer to the question of Ref.~\cite{BB19}: a lower bound of the form
\[
R_\varepsilon(f^n)
=
\Omega\!\left(
nR_{\varepsilon/n}(f)
\right)
\]
does not hold uniformly over all total Boolean functions.  For the per-distribution expected measure and the worst-case distributional and randomized measures, we obtain liminf and limsup bounds, while the existence and exact values of the corresponding limits remain open.

Beyond these asymptotic questions, determining the optimal parameter dependence in the direct-sum bound~\eqref{eq_intro_2}, as posed in Ref.~\cite{BKST24}, remains open outside the regimes covered by the present paper.

\par
Alongside direct sum theorems for the repeated relation $f^n$, composition
theorems in randomized query complexity have also received considerable
attention~\cite{BK18,BGK+20,GM21,BTLS23}. 
For compositions with the outer
function $\operatorname{OR}_n$, the AND-composition lower bound of
Ref.~\cite{GJPW17}, together with
complementation, implies
\[
R_{1/3}(\operatorname{OR}_n\circ f)
=
\Omega\bigl(nR_{1/3}(f)\bigr).
\]
Ref.~\cite{BGK+20} established the matching upper bound
\[
R_{1/3}(\operatorname{OR}_n\circ f)
=
O\bigl(nR_{1/3}(f)\bigr)
\]
without requiring error amplification for the inner algorithm.  Thus, the
standard bounded-error complexity of $\operatorname{OR}_n\circ f$ is known
up to constant factors.

The present paper considers a different, one-sided-error setting and obtains
an exact amortized characterization rather than a constant-factor estimate.
We show that the expected query complexity per input block converges to the
single-copy quantity $c_\varepsilon(f)$, which measures the expected cost of
a one-sided algorithm only on zero-inputs.  For every $\varepsilon>0$, the
same quantity also characterizes the amortized worst-case query complexity.
\subsection{Our results}
Our first main result is stated formally in \cref{Thm_asymptotics}:

\begin{Thm1}
For the expected-cost scenarios,
\begin{align}
\ED_0(f) - \varepsilon \ER_0(f)
&\leq \liminf_{n \to \infty} \frac{\overline{D}_\varepsilon^{\mu^n}(f^n)}{n}
\leq \limsup_{n \to \infty} \frac{\overline{D}_\varepsilon^{\mu^n}(f^n)}{n}
\leq  (1- \varepsilon) \ED_0(f), \label{eq_Thm1_1}\\
\lim_{n \to \infty} \frac{\overline{R}_\varepsilon(f^n)}{n}
&= (1- \varepsilon) \ER_0(f), \label{eq_Thm1_2}
\end{align}
hold for any total relation $f \subseteq \mathcal{X} \times \mathcal{Z}$ and any $\varepsilon \in [0, 1]$.
For the worst-case scenarios,
\begin{align}
 \ED_0(f)  - \varepsilon \ER_0(f)
&\leq  \liminf_{n \to \infty} \frac{D_\varepsilon^{\mu^n}(f^n)}{n}
\leq  \limsup_{n \to \infty} \frac{D_\varepsilon^{\mu^n}(f^n)}{n}
\leq \ED_0(f),\label{align_Intro_1}\\
(1- \varepsilon) \ER_0(f)
&\leq  \liminf_{n \to \infty} \frac{R_\varepsilon(f^n)}{n}
\leq  \limsup_{n \to \infty} \frac{R_\varepsilon(f^n)}{n}
\leq \ER_0(f), \label{align_Intro_2}
\end{align}
hold for any $f$ and any $\varepsilon \in (0, 1]$.
\end{Thm1}

The exact convergence statement in \eqref{eq_Thm1_2} is the main theorem of the paper: the amortized expected randomized query complexity of $f^n$ is exactly $(1-\varepsilon)\ER_0(f)$.
The per-distribution and worst-case statements in \eqref{eq_Thm1_1}, \eqref{align_Intro_1}, and \eqref{align_Intro_2} are stated as liminf/limsup bounds rather than as convergence theorems. 
The two endpoints in~\eqref{eq_Thm1_1} coincide whenever $\ED_0(f) \cong \ER_0(f)$ (which must exist due to the Yao's lemma) in which case their common value is $(1 -\varepsilon) \ER_0(f)$, but they leave open the general convergence question for these auxiliary measures.

For partial Boolean functions, if one ignores constant factors, related amortized consequences can be extracted from earlier direct-product and strong direct-sum theorems such as Refs.~\cite{Dru11,BB19,BB25}. The present work is aimed at the exact constant and at the more general setting of relations. In particular, \cref{Thm_asymptotics} identifies the precise multiplicative factor $1-\varepsilon$, applies to total relations rather than only Boolean functions, and remains meaningful for all $\varepsilon\in[0,1]$, including error parameters above $1/2$ that naturally arise for non-Boolean functions and relations.

 \cref{Thm_asymptotics} directly implies an answer to the question posed in~\cite{BB19}; \cref{Thm_asymptotics} shows that $R_\varepsilon(f^n)$ is asymptotically governed by $n\ER_0(f)$ (up to small error dependence), not by $nR_{\varepsilon/n}(f)$ in general. This distinction is necessary because there exist functions for which $\ER_0(f) \ll R_{\varepsilon/n}(f)$, as shown in Ref.~\cite{ABB+16} and \cref{Fact_appendix_cont}. 

Note additionally that the worst-case bounds in \eqref{align_Intro_1} and \eqref{align_Intro_2} are not claimed to be tight in general. The upper bounds follow from truncating repeated zero-error algorithms, while the lower bounds follow by comparison with expected cost. These two arguments match for the expected randomized measure, but they leave a gap for worst-case measures. Closing this gap, or giving examples showing that the gap is unavoidable, is left as an open problem.

\paragraph{One-sided OR compositions.}
Our second main result gives an exact amortization theorem for
$\operatorname{OR}_n\circ f$ under one-sided error.
Let $f:\mathcal{X}\to\{0,1\}$ be a Boolean partial function, and suppose that
$\mathcal{X}_0:=f^{-1}(0)$ is nonempty.  For $\varepsilon\in[0,1]$, define
\[
c_\varepsilon(f)
:=
\inf_A\max_{x\in\mathcal{X}_0}
\E_R[|A_R(x)|],
\]
where the infimum ranges over randomized query algorithms $A$ that have zero
false-positive error and false-negative error at most $\varepsilon$ on every
one-input.

As stated formally in \cref{thm:or-composition}, we prove
\[
\lim_{n\to\infty}
\frac{
\ER_{\mathrm{FP}=0,\mathrm{FN}\leq\varepsilon}
(\operatorname{OR}_n\circ f)}
{n}
=
c_\varepsilon(f)
\qquad
\text{for every }\varepsilon\in[0,1].
\]
For worst-case query complexity, we prove the corresponding identity
\[
\lim_{n\to\infty}
\frac{
R_{\mathrm{FP}=0,\mathrm{FN}\leq\varepsilon}
(\operatorname{OR}_n\circ f)}
{n}
=
c_\varepsilon(f)
\qquad
\text{for every }\varepsilon\in(0,1].
\]

Thus, both the expected-cost and worst-case amortized rates are determined by
the same single-copy quantity.  In contrast to the repeated-relation theorem,
this quantity measures the expected cost only on zero-inputs.  Positive input
blocks affect the error analysis, but their contribution to the query cost can
be reduced to a sublinear overhead.

\subsection{Separations between amortized and single-instance costs}
As consequences of \cref{Thm_asymptotics}, we derive strong separation results between amortized cost and the standard cost of a single instance.

First, in the setting of \emph{distributional query complexity}, we demonstrate an unbounded separation.
We begin by establishing the existence of a function with a large gap between its expected zero-error complexity and its worst-case error complexity:

\begin{Fact}\label{Fact_f_dist}
There exists a sequence of functions $f_m:\Bset^m \to \Bset$, distributions $\mu_m$ on $\Bset^m$, and an error parameter $0 < \varepsilon_m <1$ such that
\begin{equation*}
\overline{D}^{\mu_m}_0(f_m) = O(1) \Mand D^{\mu_m}_\varepsilon(f_m) = \Omega(m).
\end{equation*}
\end{Fact}
(The proof of \cref{Fact_f_dist} is provided in \cref{Appendix_Proof} for completeness.)

By combining \cref{Fact_f_dist} with the upper bound in \cref{Thm_asymptotics}, we obtain the following corollary, which confirms that the amortized cost can be arbitrarily smaller than the single-instance cost:

\begin{Cor}
There exists a total function $f:\Bset^m \to \Bset$, a distribution $\mu$, and a positive error $\varepsilon$ such that
\begin{equation*}
\limsup_{n \to \infty} \frac{D^{\mu^n}_\varepsilon(f^n)}{n} = O(1) \Mand D^\mu_\varepsilon(f) = \Omega(m).
\end{equation*}
\end{Cor}

For randomized query complexity, functions and relations exhibit different behavior.
For total functions, the separation is bounded polynomially due to the following known relationship between randomized and zero-error expected complexity:

\begin{Fact}[\cite{BI87,HH86,Tar89}]\label{Fact_rand_sep}
For any total function $f$ and any error $0 < \varepsilon <1$,
\begin{equation*}
R_\varepsilon(f) \leq D(f) \leq \overline{R}_0(f)^2.
\end{equation*}
\end{Fact}

This implies that for functions, the amortized cost cannot be separated from the single-instance deterministic cost by more than a quadratic factor. The following corollary follows from \cref{Thm_asymptotics} and \cref{Fact_rand_sep}:

\begin{Cor}\label{Cor_rand_square}
For any total function $f$ and any error $0 \leq \varepsilon < 1$,
\begin{equation*}
\liminf_{n \to \infty} \frac{R_\varepsilon(f^n)}{n} = \Omega(\sqrt{D(f)}).
\end{equation*}
\end{Cor}

However, for relations, no analogous polynomial bound holds, and the separation can be unbounded.
We establish this via the following existence result:

\begin{Fact}\label{Fact_f_rand}
There exists a relation $f \subseteq \Bset^m \times [m]$ and an error $\varepsilon$ such that
\begin{equation*}
\overline{R}_0(f) = \Theta(1) \Mand R_\varepsilon(f) = \Omega(m).
\end{equation*}
\end{Fact}
(The proof of \cref{Fact_f_rand} is provided in \cref{Appendix_Proof} for completeness.)
Applying the upper bound in \cref{Thm_asymptotics} to \cref{Fact_f_rand} yields an unbounded separation for relations, analogous to the distributional case:

\begin{Cor}\label{Cor_sep_relation}
There exists a total relation $f$ and an error $\varepsilon$ such that
\begin{equation*}
\limsup_{n \to \infty} \frac{R_\varepsilon(f^n)}{n} = O(1) \Mand R_\varepsilon(f) = \Omega(m).
\end{equation*}
\end{Cor}

While similar separations have been explicitly proven in communication complexity (e.g., Ref.~\cite{FKNN95}), to the best of our knowledge, these explicit separations in query complexity had not been previously addressed.
\subsection{Proof technique}
We first explain the proof of \cref{Thm_asymptotics} and then describe the
additional ideas used for the OR-composition result.  Since the worst-case
bounds in \eqref{align_Intro_1} and \eqref{align_Intro_2} follow once the
expected-cost behavior is established, we focus first on
\eqref{eq_Thm1_1} and \eqref{eq_Thm1_2}.

\subsubsection*{Upper bound strategy}
To establish the upper bounds, we use parallel repetition and truncation. For the expected-cost upper bound, we run $n$ independent copies of an optimal zero-error expected algorithm with probability $1-\varepsilon$ and otherwise output without making queries. For the worst-case upper bound, we run the repeated zero-error algorithm but terminate once a global query budget is exceeded. Chebyshev's inequality controls the probability that this truncation changes the output.

\subsubsection*{Lower bound strategy}
The proof of the lower bound for distributional expected complexity,
\begin{equation}\label{eq_tech_lower}
\liminf_{n \to \infty} \frac{\overline{D}_\varepsilon^{\mu^n}(f^n)}{n} \geq \overline{D}^\mu_0(f) - \varepsilon \ER_0(f),
\end{equation}
is more involved. We employ a coordinate-wise simulation argument: given an algorithm $\mathcal{A}$ for $f^n$, we construct an algorithm for a single instance $f$ by embedding that instance into a ``good'' coordinate $i^*$ of $\mathcal{A}$.

A technical challenge arises because the simulated single-instance algorithm cannot simultaneously have vanishing error $o(1)$, no aborts, and cost roughly $\overline{D}_\varepsilon^{\mu^n}(f^n)/n$. If such an algorithm existed, it would imply the stronger lower bound $\overline{D}^\mu_0(f)$, contradicting the upper bound $(1-\varepsilon)\overline{D}^\mu_0(f)$ in general.

The simulation instead produces a single-copy algorithm with vanishing error, expected cost roughly $\overline{D}_\varepsilon^{\mu^n}(f^n)/n$, and abort probability at most $\varepsilon$. The abort event is necessary because the simulator can verify whether earlier generated coordinates were solved correctly and can abort when one of them fails.

To relate this algorithm with aborts back to standard zero-error complexity, we use two ingredients. First, \cref{Lem_abort} shows that abort probability $\delta$ can be charged against $\ER_0(f)$, yielding a bound of the form
\[
\overline{D}^{\mu}_{\eta,0}(f)-\delta \ER_0(f)
\leq \overline{D}^{\mu}_{\eta,\delta}(f).
\]
The appearance of $\ER_0(f)$ is essential because the input distribution conditioned on aborting need not be $\mu$. Second, \cref{Lem_conti_0_CD} supplies the continuity needed to take the simulation error $\eta$ to zero. Together these tools imply \eqref{eq_tech_lower}. Abort-based arguments also appear in Ref.~\cite{BB19}.

\paragraph{Zero-error Yao's minimax theorem.}
Finally, to extend the distributional lower bound to randomized expected complexity, we require a transition principle analogous to Yao's minimax theorem. The standard bounded-error version does not directly handle expected query cost in the form needed here. We therefore prove a zero-error expected version:
\begin{equation}
\overline{R}_0(f) = \sup_{\mu} \overline{D}^\mu_0(f).
\end{equation}
This minimax relation allows us to choose a distribution $\mu$ for which $\overline{D}^\mu_0(f)$ is arbitrarily close to $\ER_0(f)$, turning the lower bound \eqref{eq_tech_lower} into $(1-\varepsilon)\ER_0(f)$.

\subsubsection*{One-sided OR-composition strategy}
The lower bound in \cref{thm:or-composition} also uses a coordinate-embedding
argument, but it does not require aborts or a limiting argument.  We embed a
single input into a uniformly random coordinate and fill every other
coordinate using an arbitrary distribution supported on zero-inputs.
A minimax theorem over such distributions then gives the exact lower bound
\[
\ER_{\mathrm{FP}=0,\mathrm{FN}\leq\varepsilon}
(\operatorname{OR}_n\circ f)
\geq n c_\varepsilon(f)
\]
for every $n$.

For the upper bound, we first choose a random set of $n^{2/3}$ coordinates
and evaluate those coordinates exactly.  If this inspection finds no
one-input, we run independent copies of a nearly optimal single-copy
one-sided algorithm on the remaining coordinates.  Random inspection detects
inputs containing many one-input blocks with high probability, while inputs
containing few one-input blocks contribute only a sublinear additional cost.
This gives an expected query cost of
\[
n c_\varepsilon(f)+o(n).
\]
For the worst-case result, we apply the same construction with a slightly
smaller false-negative parameter and truncate it at an appropriate query
budget.  Concentration and the continuity of $c_\varepsilon(f)$ in
$\varepsilon$ show that the truncation changes the false-negative probability
by only $o(1)$.

\subsection{Open problems}
We conclude by listing several open problems arising from our work.

First, \cref{Thm_asymptotics} proves that
\[
\lim_{n\to\infty}\frac{\ER_\varepsilon(f^n)}{n}=(1-\varepsilon)\ER_0(f).
\]
In contrast, for the per-distribution expected measure and for worst-case randomized and distributional measures, we currently obtain only liminf/limsup bounds. It remains open whether the corresponding amortized limits always exist, and, if they do, whether they are characterized by one of the endpoint quantities in \cref{Thm_asymptotics} or by a different single-copy complexity measure.

Second, while the direct sum property for standard \emph{quantum} query complexity is well understood---specifically, $Q_\varepsilon(f^n) = \Theta(n \cdot Q_\varepsilon(f))$ as shown in Refs.~\cite{ACLT10,AMRR11}---the behavior of \emph{expected} quantum query complexity remains unresolved. In particular, determining the exact value of the limit
\[
\lim_{n \to \infty} \frac{\overline{Q}_\varepsilon(f^n)}{n}
\]
remains an open question.

Third, an analogous question persists in communication complexity. It remains open to determine the exact value of the limit $\lim_{n \to \infty} \overline{CC}_\varepsilon^{\mu^n}(f^n)/n$ for the expected communication complexity of a relation $f$ under global error. When the success criterion is relaxed to allow \emph{coordinate-wise error}, the amortized cost is known to converge to the information complexity of $f$~\cite{BR11,Bra17}; the global-error analogue remains less well understood.
Finally, \cref{thm:or-composition} gives an exact characterization for
$\operatorname{OR}_n\circ f$ under one-sided error.  It remains open to obtain
analogous characterizations for more general outer functions, such as
$\mathsf{MAJ}_n$, and for OR compositions under two-sided error. 
\subsection*{Organization}
The remainder of this paper is organized as follows.
\cref{sec_Preliminaries} introduces the necessary preliminaries, notations, and basic properties of query algorithms with aborts used throughout the paper.
\cref{sec_continuity} establishes the continuity properties of expected distributional query complexity needed for the simulation argument.
\cref{sec_lower-bounds} derives the lower bounds and proves the zero-error Yao's minimax theorem.
\cref{sec_upper-bounds} provides the corresponding upper bounds based on parallel repetition and truncation.
\cref{sec_Main-results} synthesizes these results to prove
\cref{Thm_asymptotics}.
\cref{sec:or-composition} establishes the one-sided amortization theorem for
$\operatorname{OR}_n\circ f$, including both expected-cost and worst-case
characterizations.
Finally, \cref{Appendix_Proof} proves the existence of functions and relations
exhibiting the separations stated in \cref{Fact_f_dist,Fact_f_rand}.

\section{Preliminaries}\label{sec_Preliminaries}

\subsection{Query model and product relations}

We use the standard decision-tree model of query complexity~\cite{BdW02}.
A deterministic query algorithm is a decision tree, and a randomized query
algorithm is a distribution over deterministic query algorithms.

Let $f\subseteq\mathcal{X}\times\mathcal{Z}$ be a total relation, where
$\mathcal{X}\subseteq\Bset^m$ and $\mathcal{Z}$ is finite. Thus, for every
$x\in\mathcal{X}$, the set
\[
f(x):=\{z\in\mathcal{Z}:(x,z)\in f\}
\]
is nonempty. An algorithm succeeds on input $x$ if it outputs an element of
$f(x)$.

For $n\in\mathbb{N}$, the $n$-fold product relation
$f^n\subseteq\mathcal{X}^n\times\mathcal{Z}^n$ is defined by
\[
f^n(x_1,\ldots,x_n)
:=f(x_1)\times\cdots\times f(x_n).
\]
(Here the product $\times$ denotes Cartesian product.)
Accordingly, an algorithm computes $f^n$ on $(x_1,\ldots,x_n)$ if it outputs
$(z_1,\ldots,z_n)$ with $z_i\in f(x_i)$ for every $i\in[n]$.
For a distribution $\mu$ on $\mathcal{X}$, we write $\mu^n$ for its $n$-fold
product distribution on $\mathcal{X}^n$.

For a possibly randomized query algorithm $\calA$, let
$|\calA_r(x)|$ denote the number of queries made by $\calA$ on input $x$ when
its internal randomness is fixed to $r$. We define
\[
\Cost(\calA):=\max_{x,r}|\calA_r(x)|.
\]
For a distributional worst-case quantity under an input distribution $\mu$,
the maximum over $x$ is restricted to $\operatorname{supp}(\mu)$; for a
randomized worst-case quantity, it is taken over all inputs.

\subsection{Error and abort conventions}

An algorithm may output either an element of $\mathcal{Z}$ or a distinguished
abort symbol. For an input distribution $\mu$, let
$[f,\mu,\varepsilon,\delta]$ be the set of possibly randomized query algorithms
$\calA$ satisfying
\begin{align*}
\Pr_{x\sim\mu,r}
  [\calA_r(x)\neq\mathsf{abort}\ \text{and}\ \calA_r(x)\notin f(x)]
  &\leq\varepsilon,\\
\Pr_{x\sim\mu,r}[\calA_r(x)=\mathsf{abort}]
  &\leq\delta.
\end{align*}
The first probability is the distributional error probability. In particular,
an abort is not counted as an error.

Similarly, let $[f,\varepsilon,\delta]$ be the set of randomized query
algorithms $\calA$ such that, for every $x\in\mathcal{X}$,
\begin{align*}
\Pr_r
  [\calA_r(x)\neq\mathsf{abort}\ \text{and}\ \calA_r(x)\notin f(x)]
  &\leq\varepsilon,\\
\Pr_r[\calA_r(x)=\mathsf{abort}]
  &\leq\delta.
\end{align*}
We write
\[
[f,\mu,\varepsilon]:=[f,\mu,\varepsilon,0]
\quad\text{and}\quad
[f,\varepsilon]:=[f,\varepsilon,0].
\]
Thus, the algorithms in $[f,\varepsilon]$ never abort and have error at most
$\varepsilon$ on every input.

The success, error, and abort events are pairwise disjoint. The success
probability under $\mu$ is
\[
\Pr_{x\sim\mu,r}
  [\calA_r(x)\neq\mathsf{abort}\ \text{and}\ \calA_r(x)\in f(x)],
\]
and the analogous definition applies to the success probability on a fixed
input.

\subsection{Complexity measures}

Let $D^\mu_\varepsilon(f)$ denote the minimum worst-case query cost of a
deterministic decision tree whose error under $\mu$ is at most $\varepsilon$.
Let $R_\varepsilon(f)$ denote the minimum worst-case query cost of a randomized
query algorithm whose error is at most $\varepsilon$ on every input. Algorithms
in these definitions do not abort.

For a possibly randomized query algorithm $\calA$, its expected cost under
$\mu$ is
\[
\E_{\mu,r}[|\calA|]
:=\sum_{x\in\mathcal{X}}\mu(x)\E_r[|\calA_r(x)|].
\]
We omit the subscripts when they are clear from the context. We define
\begin{equation*}
\overline{D}^{\mu}_{\varepsilon,\delta}(f):=\inf_{\calA\in[f,\mu,\varepsilon,\delta]}
  \E_{\mu,r}[|\calA|], \quad
\overline{R}_{\varepsilon}(f):=\inf_{\calA\in[f,\varepsilon]}
  \max_{x\in\mathcal{X}}\E_r[|\calA_r(x)|].
\end{equation*}
We abbreviate
$\overline{D}^{\mu}_{\varepsilon,0}(f)$ as
$\overline{D}^{\mu}_{\varepsilon}(f)$.

In the definition of $\overline{D}^{\mu}_{\varepsilon,\delta}(f)$, both the
error probability and the expected cost are taken over the input
$x\sim\mu$ and the internal randomness of the algorithm. In particular,
$\overline{D}^{\mu}_{\varepsilon,\delta}(f)$ allows randomized decision trees;
this quantity is sometimes denoted
$\overline{R}^{\mu}_{\varepsilon,\delta}(f)$. At positive error, it may differ
from the corresponding minimum over deterministic decision trees because
randomization can trade error probability for expected cost. At zero error,
every deterministic decision tree assigned positive probability by a
zero-error randomized algorithm has zero error under $\mu$. Since expected cost
is linear in the distribution over decision trees, randomization does not
change the distributional optimum at zero error.

The infima in the definitions of expected complexity need not be attained.
Whenever an optimal expected-cost algorithm is mentioned, it may be replaced
by an algorithm whose expected cost is arbitrarily close to the corresponding
infimum.

We also define the standard deviation of the query cost by
\begin{equation*}
\sigma(\calA,\mu):=\sqrt{\E_{\mu,r}[|\calA|^2]-\E_{\mu,r}[|\calA|]^2}, \quad
\E_{\mu,r}[|\calA|^2]:=\sum_{x\in\mathcal{X}}\mu(x)
  \E_r[|\calA_r(x)|^2].
\end{equation*}

Table~\ref{tab:notation} summarizes the notation used throughout the paper.
\begin{table}[t]
\centering
\renewcommand{\arraystretch}{1.2}
\begin{tabular}{p{0.30\textwidth}p{0.62\textwidth}}
\hline
\textbf{Notation} & \textbf{Meaning} \\
\hline

$|\calA_r(x)|$
& Number of queries made by $\calA$ on input $x$ with randomness $r$ \\

$\Cost(\calA)$
& Worst-case number of queries over the relevant inputs and randomness \\

$\E_{\mu,r}[|\calA|]$
& Expected number of queries under the input distribution $\mu$ \\

$\sigma(\calA,\mu)$
& Standard deviation of the number of queries under $\mu$ \\

$[f,\mu,\varepsilon,\delta]$
& Algorithms with distributional error $\leq\varepsilon$ and abort probability
  $\leq\delta$ under $\mu$ \\

$[f,\varepsilon,\delta]$
& Algorithms with error $\leq\varepsilon$ and abort probability $\leq\delta$
  on every input \\

$D^\mu_\varepsilon(f)$
& Distributional query complexity with worst-case cost \\

$R_\varepsilon(f)$
& Randomized query complexity with worst-case cost \\

$\overline{D}^{\mu}_{\varepsilon,\delta}(f)$
& Distributional expected query complexity with error $\varepsilon$ and abort
  probability $\delta$ \\

$\overline{D}^{\mu}_{\varepsilon}(f)$
& Distributional expected query complexity with error $\varepsilon$ and no
  aborts \\

$\overline{R}_\varepsilon(f)$
& Randomized expected query complexity with worst-case error $\varepsilon$ \\

\hline
\end{tabular}
\caption{Notation used throughout the paper.}
\label{tab:notation}
\end{table}

\subsection{Basic tools}

\begin{Fact}\label{Fact_conti_expe}
For any query algorithm $\calA$ on $\mathcal{X}$, the maximum values $\max_{\mu} \E_\mu[|\calA|]$ and $\max_\mu \sigma(\calA, \mu)$ exist.
\end{Fact}
\begin{proof}
Both $\E_\mu[|\calA|]$ and $\sigma(\calA, \mu)$ are continuous on $\mu$, as follows directly from their finite-sum expressions. Thus, the maximum values exist, following the standard fact that any continuous function on a compact set achieves a finite maximum.

\end{proof}

The lower-bound simulation later produces a single-copy algorithm with small error and nonzero abort probability. The following lemma explains how abort probability is charged in terms of $\ER_0(f)$.

\begin{Lem}\label{Lem_abort}
\begin{equation*}
\ED_{\varepsilon, 0}(f) - \delta \ER_{0}(f)\leq \ED_{\varepsilon, \delta}(f).
\end{equation*}
\end{Lem}

\begin{proof}
Let \(\calA_1\in[f,\mu,\varepsilon,\delta]\) be an algorithm with
\[
  \E_{\mu}[|\calA_1|]\le
  \overline D^\mu_{\varepsilon,\delta}(f)+\eta .
\]
If \(\calA_1\) aborts with probability \(0\), there is nothing to prove.
Otherwise, let \(\nu\) denote the distribution of the input conditioned
on the event that \(\calA_1\) aborts.

Now choose a zero-error randomized algorithm \(\calA_2\in[f,0]\) satisfying
\[
  \max_x \E[|\calA_2(x)|]\le \overline R_0(f)+\eta .
\]
In particular,
\[
  \E_{x\sim\nu}[|\calA_2(x)|]\le \overline R_0(f)+\eta .
\]
Run \(\calA_1\), and whenever \(\calA_1\) aborts, continue by running \(\calA_2\).
The resulting algorithm has no aborts and has error at most
\(\varepsilon\), since \(\calA_2\) is zero-error. Its expected cost is at most
\[
  \overline D^\mu_{\varepsilon,\delta}(f)+\eta
  + \delta(\overline R_0(f)+\eta).
\]
Letting \(\eta\to0\) gives
\[
  \overline D^\mu_{\varepsilon,0}(f)-\delta\overline R_0(f)
  \le
  \overline D^\mu_{\varepsilon,\delta}(f).
\]

The reason \(\overline R_0(f)\) appears, rather than
\(\overline D^\mu_0(f)\), is that the distribution \(\nu\) of inputs
conditioned on aborting is not controlled and need not equal \(\mu\).
The continuation algorithm must therefore work with small expected cost
for every possible abort-conditioned distribution.
\end{proof}

\section{Continuity}\label{sec_continuity}
\cref{sec_continuity} proves continuity properties of the quantity $\overline{D}^\mu_{\varepsilon,\delta}(f)$ with respect to the error parameter $\varepsilon$.
The main theorem only uses the zero-error continuity statement, \cref{Lem_conti_0_CD}; \cref{Lem_conti_CD} records the analogous positive-error continuity for completeness.
\begin{Lem}\label{Lem_conti_CD}
For any $\varepsilon > 0$ and $\delta \geq 0$,
$\lim_{\rho \to \varepsilon}\ED_{\rho, \delta}(f)
=\ED_{\varepsilon, \delta}(f).$
\end{Lem}
\begin{proof}
Since both limits $\rho \searrow \varepsilon$ and $\rho \nearrow \varepsilon$ are proved similarly, we only show the case $\rho \searrow \varepsilon$.
Take $\calA \in [f, \mu, \rho, \delta]$ and create a new algorithm $\tilde{\calA} \in [f, \mu, \varepsilon, \delta]$ as running $\calA$ w.p. $ \varepsilon' := \varepsilon / (2\rho - \varepsilon)$ and another algorithm $\calA' \in [f, \mu, \varepsilon/2, \delta]$ w.p. $1 - \varepsilon'$.
This construction interpolates between error $\rho$ and error $\varepsilon$ by randomization, while preserving expected cost linearly.
\par
The expectation of this algorithm satisfies
\begin{equation*}
\E[|\tilde{\calA}|]
= \varepsilon' \E[|\calA|] + (1- \varepsilon') \E[|\calA'|].
\end{equation*}
Let $\calA \in [f, \mu, \rho, \delta]$ be an optimal algorithm, i.e.,  $\E[|\calA|] = \ED_{\rho, \delta}(f)$.
Then the above equality implies
\begin{equation*}
\ED_{\varepsilon, \delta}(f)
\leq
\E[|\tilde{\calA}|]
=  \frac{\varepsilon}{2\rho - \varepsilon} \ED_{\rho, \delta}(f) + \left(1-\frac{\varepsilon}{2\rho - \varepsilon}\right) \E[|\calA'|].
\end{equation*}
In addition, $\ED_{\rho, \delta}(f) \leq \ED_{\varepsilon, \delta}(f)$ trivially holds.
Therefore, taking $\rho \searrow \varepsilon$ yields the desired statement.
\end{proof}

\begin{Lem}\label{Lem_conti_0_CD}
Suppose $\delta, \delta' \geq 0$. Then for any $f \subseteq \mathcal{X} \times \mathcal{Z}$,
\begin{equation*}
\ED_{\alpha, \delta+ \sqrt{\varepsilon} \delta'}(f) \leq \ED_{\varepsilon, \delta}(f) + \sqrt{\varepsilon} ~\ED_{\alpha/\sqrt{\varepsilon}, \delta'}(f)
\end{equation*}
holds for any positive $\varepsilon$ satisfying
$ \sqrt{\varepsilon} < \mu_\textsf{min} := \min_{x \in \operatorname{supp}(\mu)} \mu(x)$, and any $\alpha \in [0, \sqrt{\varepsilon})$.

In particular the case of $\alpha = \delta' = 0$ shows that $\ED_{\varepsilon, \delta}(f)$ is also continuous at $\varepsilon =0$.
\end{Lem}
\begin{proof}
Fix an arbitrary $\eta>0$. Choose algorithms
$\calA \in [f,\mu,\varepsilon,\delta]$ and
$\calA' \in [f,\mu,\alpha/\sqrt{\varepsilon},\delta']$ satisfying
\begin{equation*}
\E_{\mu}[|\calA|]
\leq \ED_{\varepsilon,\delta}(f)+\eta, \quad
\E_{\mu}[|\calA'|]
\leq \ED_{\alpha/\sqrt{\varepsilon},\delta'}(f)+\eta.
\end{equation*}
View $\calA$ as first sampling a deterministic decision tree $T$ according to its internal randomness and then running $T$ on the input. A deterministic tree may have ordinary output leaves and abort leaves. 
Call $T$ \emph{bad} if there exists an input $x\in\operatorname{supp}(\mu)$ such that $T(x)$ does not abort and
\begin{equation*}
T(x)\notin f(x).
\end{equation*}
Otherwise, call $T$ \emph{good}. 
Note that any query algorithm may decide whether $T$ is good or bad without any additional query access.
Thus, every good tree is correct on every input in $\operatorname{supp}(\mu)$ whenever it does not abort.

Let $\mathsf{Bad}$ denote the event that the tree sampled by $\calA$ is bad. For every bad tree $T$, choose an input $x_T\in\operatorname{supp}(\mu)$ on which $T$ outputs incorrectly. Since $\mu(x_T)\geq\mu_{\textsf{min}}$, the distributional error of $\calA$ satisfies
\begin{align*}
\varepsilon
&\geq
\Pr_{x\sim\mu,T}\bigl[T(x)\text{ does not abort and }T(x)\notin f(x)\bigr]\\
&=
\E_T\left[
\sum_{x\in\operatorname{supp}(\mu)}
\mu(x)\mathbf{1}\bigl[T(x)\text{ does not abort and }T(x)\notin f(x)\bigr]
\right]\\
&\geq
\mu_{\textsf{min}}\Pr_T(\mathsf{Bad}).
\end{align*}
Consequently,
\begin{equation}\label{eq_bad_tree_probability}
\Pr_T(\mathsf{Bad})
\leq \frac{\varepsilon}{\mu_{\textsf{min}}}
<\sqrt{\varepsilon}.
\end{equation}

We now define a new algorithm $\widetilde{\calA}$ for $f$.

\begin{center}
\resizebox{0.9\textwidth}{!}{
\begin{tabular}{|ll|} \hline
\multicolumn{2}{|c|}{\textbf{A new algorithm $\widetilde{\calA}$} for~\cref{Lem_conti_0_CD}}\\ \hline
1. &Sample the deterministic decision tree $T$ used by $\calA$.\\
2. &If $T$ is good, run $T$ on the input and return its output or abort whenever $T$ aborts.\\
3. &If $T$ is bad, discard $T$, run $\calA'$, and return the output or abort whenever $\calA'$ aborts.\\ \hline
\end{tabular}
}
\end{center}

The event $\mathsf{Bad}$ depends only on the sampled tree $T$ and is therefore independent of the input $x\sim\mu$. On the event $\mathsf{Bad}^{c}$, the algorithm is zero-error on $\operatorname{supp}(\mu)$. Hence the distributional error of $\widetilde{\calA}$ is at most
\begin{equation*}
\Pr_T(\mathsf{Bad})
\Pr_{x\sim\mu,R'}\bigl[\calA'(x)\text{ does not abort and }\calA'(x)\notin f(x)\bigr]
\leq
\sqrt{\varepsilon}\cdot\frac{\alpha}{\sqrt{\varepsilon}}
=\alpha,
\end{equation*}
where we used \eqref{eq_bad_tree_probability}.

Similarly, the abort probability of $\widetilde{\calA}$ is at most
\begin{equation*}
\Pr_{x\sim\mu,T}\bigl[\mathsf{Bad}^{c}\text{ and }T(x)\text{ aborts}\bigr]
+
\Pr_T(\mathsf{Bad})
\Pr_{x\sim\mu,R'}\bigl[\calA'(x)\text{ aborts}\bigr]\\
\leq
\delta+\sqrt{\varepsilon}\,\delta'.
\end{equation*}
Therefore,
$\widetilde{\calA}\in
[f,\mu,\alpha,\delta+\sqrt{\varepsilon}\,\delta'].$

Finally, its expected query cost satisfies
\begin{align*}
\E_{\mu}[|\widetilde{\calA}|]
&=
\E_{x\sim\mu,T}\left[
\mathbf{1}_{\mathsf{Bad}^{c}}|T(x)|
\right]
+
\Pr_T(\mathsf{Bad})\E_{\mu}[|\calA'|]\\
&\leq
\E_{\mu}[|\calA|]
+
\sqrt{\varepsilon}\E_{\mu}[|\calA'|]\\
&\leq
\ED_{\varepsilon,\delta}(f)
+
\sqrt{\varepsilon}\ED_{\alpha/\sqrt{\varepsilon},\delta'}(f)
+(1+\sqrt{\varepsilon})\eta.
\end{align*}
Since $\eta>0$ was arbitrary, we obtain
\begin{equation*}
\ED_{\alpha,\delta+\sqrt{\varepsilon}\delta'}(f)
\leq
\ED_{\varepsilon,\delta}(f)
+
\sqrt{\varepsilon}\ED_{\alpha/\sqrt{\varepsilon},\delta'}(f).
\end{equation*}

For the final assertion, set $\alpha=\delta'=0$. Monotonicity in the error parameter and the preceding inequality give
\begin{equation*}
\ED_{\varepsilon,\delta}(f)
\leq
\ED_{0,\delta}(f)
\leq
\ED_{\varepsilon,\delta}(f)
+
\sqrt{\varepsilon}\ED_{0,0}(f).
\end{equation*}
Since $f$ is a total relation with finite input length, $\ED_{0,0}(f)<\infty$. Letting $\varepsilon\searrow0$ proves continuity at zero.
\end{proof}

\section{Proof of lower bounds}\label{sec_lower-bounds}
The goal of \cref{sec_lower-bounds} is to prove the lower bounds for
\begin{equation*}
\liminf_{n \to \infty} \frac{\overline{D}^{\mu^n}_\varepsilon(f^n)}{n}
\Mand
\liminf_{n \to \infty} \frac{\overline{R}_\varepsilon(f^n)}{n}.
\end{equation*}
We first prove the per-distribution lower bound in \cref{Prop_Lower-bound}. We then combine it with the zero-error expected minimax theorem in \cref{Lem_Yao_exp} to obtain the randomized lower bound in \cref{Prop_Lower-bound_R}.

\begin{Prop}\label{Prop_Lower-bound}
For $\varepsilon \in [0, 1)$,
\begin{equation*}
\ED_{0, 0}(f) - \varepsilon \ER_0(f) \leq \liminf_{n \to \infty} \frac{\overline{D}^{\mu^n}_\varepsilon(f^n)}{n}.
\end{equation*}
\end{Prop}

\begin{proof}
Let $\calA_R$ be an algorithm for $[f^n,\mu^n,\varepsilon]$ with
\[
\E_{x^n\sim\mu^n,R}[|\calA_R(x^n)|]
\leq \overline{D}^{\mu^n}_\varepsilon(f^n)+\eta_n,
\]
where $\eta_n=o(n)$. Throughout this proof $x^n=(x_1,\ldots,x_n)$ is sampled from $\mu^n$, and the probability also includes the internal randomness $R$ of the algorithm.

For each coordinate $i$, let
\[
S_i := \{\calA_R(x^n)_i \notin f(x_i)\}
\Mand
G_{<i}:=\bigcap_{j<i} S_j^c.
\]
The events $S_i\cap G_{<i}$ are disjoint and their union is the event that at least one coordinate is wrong. Hence
\begin{equation*}
\varepsilon \geq \Pr(\calA_R(x^n)\notin f^n(x^n))
=\sum_{i\le n}\Pr(S_i\cap G_{<i}).
\end{equation*}
For any $\alpha>1$, Markov's inequality implies that at least $\lceil(1-1/\alpha)n\rceil$ coordinates satisfy
\begin{equation}\label{Prop_Lb_eq_1}
\Pr(S_i\cap G_{<i})\leq \alpha\frac{\varepsilon}{n}.
\end{equation}

Let $q_i(x^n,r)$ be the number of queries that $\calA_r$ makes to the $i$th input block. Since $\sum_i q_i(x^n,r)=|\calA_r(x^n)|$, for any $\beta>1$ Markov's inequality implies that at least $\lceil(1-1/\beta)n\rceil$ coordinates satisfy
\begin{equation*}
\E_{x^n\sim\mu^n,R}[q_i(x^n,R)]
\leq \beta\frac{\E_{x^n\sim\mu^n,R}[|\calA_R(x^n)|]}{n}.
\end{equation*}
Set $\alpha=2\beta/(\beta-1)$. The two large sets of coordinates intersect, so there is a coordinate $i^*$ satisfying both \eqref{Prop_Lb_eq_1} and
\begin{equation*}
\E_{x^n\sim\mu^n,R}[q_{i^*}(x^n,R)]
\leq \beta\frac{\E_{x^n\sim\mu^n,R}[|\calA_R(x^n)|]}{n}.
\end{equation*}

We now construct a single-copy algorithm $\calA_{\mathsf{new}}$ for $f$.
\begin{center}
\resizebox{1.0\textwidth}{!}{
\begin{tabular}{|ll|} \hline
\multicolumn{2}{|c|}{\textbf{Algorithm $\calA_{\mathsf{new}}$}}\\ \hline
1.&Sample $z_j\sim\mu$ independently for every $j\neq i^*$.\\
2.&Simulate $\calA_R$ on $y^n=(z_1,\ldots,z_{i^*-1},x,z_{i^*+1},\ldots,z_n)$.\\
3.&Abort if $\calA_R(y^n)_j\notin f(z_j)$ for some $j<i^*$. This check uses only sampled coordinates.\\
4.&If no abort occurs, output $\calA_R(y^n)_{i^*}$.\\ \hline
\end{tabular}
}
\end{center}
The distribution of $y^n$ is exactly $\mu^n$. Therefore, by the choice of $i^*$,
\begin{equation}\label{eq_Prop_L_1}
\Pr(\calA_{\mathsf{new}}(x)\notin f(x) \land \text{$\calA_{\mathsf{new}}$ does not abort})
\leq \frac{\alpha\varepsilon}{n}
\end{equation}
and
\begin{equation*}
\E_{x\sim\mu,R}[|\calA_{\mathsf{new},R}(x)|]
\leq
\E_{x^n\sim\mu^n,R}[q_{i^*}(x^n,R)]
\leq \beta\frac{\E_{x^n\sim\mu^n,R}[|\calA_R(x^n)|]}{n}.
\end{equation*}
Moreover,
\begin{equation*}
\Pr(\text{$\calA_{\mathsf{new}}$ aborts})
=\Pr(G_{<i^*}^c)
\leq \Pr(\calA_R(x^n)\notin f^n(x^n))
\leq \varepsilon.
\end{equation*}
Combining this abort bound with \eqref{eq_Prop_L_1}, the unconditional error probability of $\calA_{\mathsf{new}}$ is at most $\alpha\varepsilon/n$. Thus
\begin{equation*}
\ED_{\alpha\frac{\varepsilon}{n},\varepsilon}(f)
\leq \E_{x\sim\mu,R}[|\calA_{\mathsf{new},R}(x)|]
\leq \beta\frac{\overline{D}^{\mu^n}_\varepsilon(f^n)+\eta_n}{n}.
\end{equation*}
Applying \cref{Lem_abort} gives
\begin{equation*}
\ED_{\alpha\frac{\varepsilon}{n},0}(f)-\varepsilon\ER_0(f)
\leq \beta\frac{\overline{D}^{\mu^n}_\varepsilon(f^n)+\eta_n}{n}.
\end{equation*}
For $\varepsilon=0$, this already yields the claim after taking $n\to\infty$, then $\beta\downarrow1$. For $0<\varepsilon<1$, take $n\to\infty$ and use the zero-error continuity statement \cref{Lem_conti_0_CD}; since $\eta_n=o(n)$, we obtain
\begin{equation*}
\ED_{0,0}(f)-\varepsilon\ER_0(f)
\leq \beta\liminf_{n\to\infty}\frac{\overline{D}^{\mu^n}_\varepsilon(f^n)}{n}.
\end{equation*}
Finally let $\beta\downarrow1$.
\end{proof}

\cref{Lem_Yao_exp} is a zero-error version of Yao's minimax theorem for expected query complexity. The zero-error condition is important because it allows randomized algorithms to be viewed as distributions over deterministic algorithms that are correct on all relevant inputs.
\begin{Lem}\label{Lem_Yao_exp}
$\ER_0(f) = \sup_\mu \ED_0(f).$
\end{Lem}
\begin{proof}
We prove the two inequalities separately.

\textbf{Direction 1: $\ER_0(f) \geq \sup_\mu \ED_0(f)$.}
Fix $\gamma>0$. By definition of $\ER_0(f)$, there is a randomized zero-error algorithm, equivalently a distribution $\rho$ over deterministic zero-error decision trees $\{\calA_r\}$, such that
\begin{equation*}
\max_x \E_{r\sim\rho}[|\calA_r(x)|]
\leq \ER_0(f)+\gamma.
\end{equation*}
For every input distribution $\nu$,
\begin{equation*}
\E_{x\sim\nu,r\sim\rho}[|\calA_r(x)|]
\leq \max_x \E_{r\sim\rho}[|\calA_r(x)|]
\leq \ER_0(f)+\gamma.
\end{equation*}
The same randomized algorithm belongs to $[f,\nu,0]$, and therefore
\begin{equation*}
\overline{D}^{\nu}_0(f)
\leq \E_{x\sim\nu,r\sim\rho}[|\calA_r(x)|]
\leq \ER_0(f)+\gamma.
\end{equation*}
Taking the supremum over $\nu$ and then letting $\gamma\to0$ proves this direction.

\textbf{Direction 2: $\ER_0(f) \leq \sup_\mu \ED_0(f)$.}
We view a randomized zero-error algorithm as a distribution over the finite set of deterministic zero-error decision trees.\footnote{For this purpose, without loss of generality we restrict deterministic algorithms for a total relation $f\subseteq\Bset^m\times\mathcal Z$ to at most $m$ queries; querying all input bits is always sufficient.}
Von Neumann's minimax theorem gives
\begin{equation*}
\ER_0(f)
=\min_{\calA_R\in[f,0]}\max_\mu \E_{\mu,R}[|\calA_R(x)|]
=\max_\mu\min_{\calA_R\in[f,0]}\E_{\mu,R}[|\calA_R(x)|].
\end{equation*}
The inner minimum is achieved at an extreme point, hence at a deterministic zero-error decision tree. Thus
\begin{equation*}
\ER_0(f)=\max_\mu \min_{\calA_D\text{: deterministic zero-error}}\E_\mu[|\calA_D(x)|].
\end{equation*}
Let $\mu^*$ be a distribution achieving the maximum, and for $0<\theta<1$ let
\begin{equation*}
\mu_\theta := (1-\theta)\mu^* + \theta\cdot\mathsf{Uni},
\end{equation*}
where $\mathsf{Uni}$ is the uniform distribution on $\mathcal X$. Then $\mu_\theta$ has full support and converges to $\mu^*$ as $\theta\to0$. Since the set of deterministic zero-error trees is finite, the function
$\mu\mapsto \min_{\calA_D}\E_\mu[|\calA_D(x)|]$ is continuous. Therefore
\begin{align}
\ER_0(f)
&= \min_{\calA_D\text{: deterministic zero-error}} \E_{\mu^*}[|\calA_D(x)|]\nonumber\\
&= \lim_{\theta\to0}\min_{\calA_D\text{: deterministic zero-error}}\E_{\mu_\theta}[|\calA_D(x)|]\label{Lem_Yao_exp_align1}\\
&\leq \sup_{\mu\text{: full support}}\min_{\calA_D\text{: deterministic zero-error}}\E_{\mu}[|\calA_D(x)|]\label{Lem_Yao_exp_align3}\\
&= \sup_{\mu\text{: full support}}\ED_0(f)\label{Lem_Yao_exp_align4}\\
&\leq \sup_{\mu\text{: all}}\ED_0(f).\label{Lem_Yao_exp_align5}
\end{align}
For full-support $\mu$, zero distributional error under $\mu$ forces correctness on every input in $\mathcal X$ and on every deterministic tree in the support of the randomized algorithm. This justifies \eqref{Lem_Yao_exp_align4} and completes the proof.
\end{proof}

\begin{Prop}\label{Prop_Lower-bound_R}
For $\varepsilon \in [0, 1)$,
\begin{equation*}
(1 - \varepsilon) \ER_{0}(f) \leq \liminf_{n \to \infty} \frac{\ER_\varepsilon(f^n)}{n}.
\end{equation*}
\end{Prop}
\begin{proof}
By \cref{Lem_Yao_exp}, for every $\gamma>0$ there is a distribution $\mu^*$ such that
\begin{equation*}
\ER_0(f)-\gamma < \overline{D}^{\mu^*}_0(f).
\end{equation*}
Applying \cref{Prop_Lower-bound} to $\mu^*$ gives
\begin{equation*}
\overline{D}^{\mu^*}_0(f)-\varepsilon\ER_0(f)
\leq \liminf_{n\to\infty}\frac{\overline{D}^{(\mu^*)^n}_\varepsilon(f^n)}{n}.
\end{equation*}
Since any randomized worst-case algorithm is also a distributional expected algorithm under $(\mu^*)^n$,
\begin{equation*}
\liminf_{n\to\infty}\frac{\overline{D}^{(\mu^*)^n}_\varepsilon(f^n)}{n}
\leq
\liminf_{n\to\infty}\frac{\ER_\varepsilon(f^n)}{n}.
\end{equation*}
Therefore
\begin{equation*}
(1-\varepsilon)\ER_0(f)-\gamma
\leq \liminf_{n\to\infty}\frac{\ER_\varepsilon(f^n)}{n}.
\end{equation*}
Letting $\gamma\to0$ proves the claim.
\end{proof}

\section{Proof of upper bounds}\label{sec_upper-bounds}
\cref{sec_upper-bounds} proves the upper bounds used in \cref{Thm_asymptotics}. The expected-cost upper bounds are obtained by running the zero-error algorithm only with probability $1-\varepsilon$. The worst-case upper bounds are obtained by repeating a zero-error algorithm and truncating the total number of queries.

\begin{Lem}\label{Lem_algo_CD}
For any $n \in \mathbb{N}$ and $\varepsilon >0$, there is an algorithm $\calA \in [f^n, \mu^n, \varepsilon]$ such that
\begin{equation}\label{eq_Lem_algo_CD}
\Cost(\calA) \leq  n\overline{D}^{\mu}_{0}(f) + o(n).
\end{equation}
Consequently, $D_\varepsilon^{\mu^n}(f^n) \leq n\overline{D}^{\mu}_{0}(f) + o(n)$.
\par
Additionally, for any $n \in \mathbb{N}$ and $\varepsilon \in[0,1]$, there is an algorithm $\calA \in [f^n, \mu^n, \varepsilon]$ such that
\begin{equation}\label{eq_Lem_algo_expCD}
\E_{\mu^n}[|\calA|] \leq  n(1-\varepsilon)\overline{D}^{\mu}_{0}(f)+o(n).
\end{equation}
Consequently, $\overline{D}_\varepsilon^{\mu^n}(f^n) \leq n(1 - \varepsilon)\overline{D}^{\mu}_{0}(f)+o(n)$.
\end{Lem}

\begin{proof}
For each $n$, choose a zero-error algorithm $\calB\in[f,\mu,0]$ satisfying
\begin{equation*}
\E_\mu[|\calB|]\leq \overline{D}^{\mu}_0(f)+\gamma_n,
\end{equation*}
where $\gamma_n\to0$. Let $\calB^n$ be the algorithm that runs independent copies of $\calB$ on all $n$ coordinates. Then $\calB^n$ has zero error under $\mu^n$ and
\begin{equation*}
\E_{\mu^n}[|\calB^n|]\leq n(\overline{D}^{\mu}_0(f)+\gamma_n).
\end{equation*}
Since the input length of $f$ is fixed, we may assume every deterministic decision tree queries at most all input bits; hence the variance of the query cost of $\calB$ is bounded by a constant depending only on $f$. Chebyshev's inequality therefore gives, for every fixed $k>0$,
\begin{equation*}
\Pr_{x^n\sim\mu^n,R}\left(|\calB^n_R(x^n)|\geq n(\overline{D}^{\mu}_0(f)+\gamma_n)+k\sqrt{n}\,\sigma(\calB,\mu)\right)
\leq \frac{1}{k^2}.
\end{equation*}
Truncate $\calB^n$ when the query count reaches this threshold. With $k=\lceil1/\sqrt{\varepsilon}\rceil$, the truncation causes error at most $\varepsilon$ under $\mu^n$, and the worst-case cost is at most $n\overline{D}^{\mu}_0(f)+o(n)$. This proves \eqref{eq_Lem_algo_CD}.

For the expected-cost statement, run $\calB^n$ with probability $1-\varepsilon$ and output an arbitrary value with probability $\varepsilon$. This algorithm has error at most $\varepsilon$ under $\mu^n$ and expected query cost at most
\begin{equation*}
(1-\varepsilon)n(\overline{D}^{\mu}_0(f)+\gamma_n)
= n(1-\varepsilon)\overline{D}^{\mu}_0(f)+o(n),
\end{equation*}
which proves \eqref{eq_Lem_algo_expCD}.
\end{proof}

\begin{Lem}\label{Lem_algo_CR}
For any $n \in \mathbb{N}$ and $\varepsilon > 0$, there is an algorithm $\calA \in [f^n, \varepsilon]$ such that
\begin{equation}\label{eq_Lem_algo_CR}
\Cost(\calA) \leq  n\ER_0(f) + o(n).
\end{equation}
Consequently, $R_\varepsilon(f^n) \leq n \ER_0(f) + o(n)$.
\par
Additionally, for any $n \in \mathbb{N}$ and $\varepsilon\in[0,1]$, there is an algorithm $\calA_R \in [f^n, \varepsilon]$ such that
\begin{equation}\label{eq_Lem_algo_expCR}
\E_R[|\calA_R(x^n)|] \leq  n(1-\varepsilon)\overline{R}_{0}(f)+o(n)
\end{equation}
for every input $x^n$.
Consequently, $\overline{R}_\varepsilon(f^n) \leq n(1 - \varepsilon)\overline{R}_{0}(f)+o(n)$.
\end{Lem}

\begin{proof}
For each $n$, choose a zero-error randomized algorithm $\calB\in[f,0]$ satisfying
\begin{equation*}
\max_x\E_R[|\calB_R(x)|]\leq \ER_0(f)+\gamma_n,
\end{equation*}
where $\gamma_n\to0$. Let $\calB^n$ be the algorithm that runs independent copies of $\calB$ on all $n$ coordinates. For any product distribution $\mu^{\otimes n}=\mu_1\times\cdots\times\mu_n$,
\begin{equation*}
\E_{\mu^{\otimes n}}[|\calB^n|]=\sum_{i\le n}\E_{\mu_i}[|\calB|]
\leq n(\ER_0(f)+\gamma_n).
\end{equation*}
The standard deviation satisfies
\begin{equation*}
\sigma_n(\calB^n,\mu^{\otimes n})=\sqrt{\sum_{i\le n}\sigma^2(\calB,\mu_i)}
\leq \sqrt n\max_\nu\sigma(\calB,\nu).
\end{equation*}
By \cref{Fact_conti_expe}, the maximum on the right exists. Chebyshev's inequality implies
\begin{equation}\label{Thm_eq_Che_random_GO}
\Pr_{\mu^{\otimes n}}\left(|\calB^n|\geq n(\ER_0(f)+\gamma_n)+k\sqrt n\max_\nu\sigma(\calB,\nu)\right)
\leq \frac1{k^2}.
\end{equation}
Define $\calB^n_k$ by truncating $\calB^n$ when the number of queries reaches the threshold in \eqref{Thm_eq_Che_random_GO}. For any fixed input $(x^0_1,\ldots,x^0_n)$, apply \eqref{Thm_eq_Che_random_GO} to the product distribution with $\mu_i(x_i^0)=1$. Then the error probability of $\calB^n_k$ on this input is at most $1/k^2$. Taking $k=\lceil1/\sqrt\varepsilon\rceil$, we obtain $\calB^n_k\in[f^n,\varepsilon]$ and
\begin{equation*}
\Cost(\calB^n_k)\leq n\ER_0(f)+o(n),
\end{equation*}
which proves \eqref{eq_Lem_algo_CR}.

For the expected-cost statement, run $\calB^n$ with probability $1-\varepsilon$ and output an arbitrary value with probability $\varepsilon$. For every input $x^n$, the error probability is at most $\varepsilon$, and the expected query cost is at most
\begin{equation*}
(1-\varepsilon)n(\ER_0(f)+\gamma_n)
= n(1-\varepsilon)\ER_0(f)+o(n).
\end{equation*}
This proves \eqref{eq_Lem_algo_expCR}.
\end{proof}

\section{Main results}\label{sec_Main-results}
We now assemble the lower bounds from \cref{sec_lower-bounds} and the upper bounds from \cref{sec_upper-bounds}.

\begin{Thm}\label{Thm_asymptotics}
For the expected-cost scenarios,
\begin{align*}
\ED_0(f) - \varepsilon \ER_0(f)
&\leq \liminf_{n \to \infty} \frac{\overline{D}_\varepsilon^{\mu^n}(f^n)}{n}
\leq \limsup_{n \to \infty} \frac{\overline{D}_\varepsilon^{\mu^n}(f^n)}{n}
\leq  (1- \varepsilon) \ED_0(f),\\
\lim_{n \to \infty} \frac{\overline{R}_\varepsilon(f^n)}{n}
&= (1- \varepsilon) \ER_0(f),
\end{align*}
hold for any total relation $f \subseteq \mathcal{X} \times \mathcal{Z}$ and any $\varepsilon \in [0, 1]$.
For the worst-case scenarios,
\begin{align*}
 \ED_0(f)  - \varepsilon \ER_0(f)
&\leq  \liminf_{n \to \infty} \frac{D_\varepsilon^{\mu^n}(f^n)}{n}
\leq  \limsup_{n \to \infty} \frac{D_\varepsilon^{\mu^n}(f^n)}{n}
\leq \ED_0(f),\\
(1- \varepsilon) \ER_0(f)
&\leq  \liminf_{n \to \infty} \frac{R_\varepsilon(f^n)}{n}
\leq  \limsup_{n \to \infty} \frac{R_\varepsilon(f^n)}{n}
\leq \ER_0(f),
\end{align*}
hold for any $f$ and any $\varepsilon \in (0, 1]$.
\end{Thm}

\begin{proof}
We first prove the expected per-distribution bounds. For $\varepsilon\in[0,1)$, \cref{Prop_Lower-bound} gives
\begin{equation*}
\ED_0(f)-\varepsilon\ER_0(f)
\leq
\liminf_{n\to\infty}\frac{\overline{D}^{\mu^n}_\varepsilon(f^n)}{n}.
\end{equation*}
For $\varepsilon=1$, the same lower bound is trivial because the left-hand side is at most zero. The upper bound follows from \cref{Lem_algo_CD}, specifically \eqref{eq_Lem_algo_expCD}:
\begin{equation*}
\limsup_{n\to\infty}\frac{\overline{D}^{\mu^n}_\varepsilon(f^n)}{n}
\leq (1-\varepsilon)\ED_0(f).
\end{equation*}
Together with the elementary inequality $\liminf\leq\limsup$, this proves the first displayed chain.

Next we prove the randomized expected identity. For $\varepsilon\in[0,1)$, \cref{Prop_Lower-bound_R} gives
\begin{equation*}
(1-\varepsilon)\ER_0(f)
\leq
\liminf_{n\to\infty}\frac{\ER_\varepsilon(f^n)}{n}.
\end{equation*}
For $\varepsilon=1$, the lower bound is again trivial. The upper bound follows from \cref{Lem_algo_CR}, specifically \eqref{eq_Lem_algo_expCR}:
\begin{equation*}
\limsup_{n\to\infty}\frac{\ER_\varepsilon(f^n)}{n}
\leq (1-\varepsilon)\ER_0(f).
\end{equation*}
The lower and upper bounds match, so the limit exists and equals $(1-\varepsilon)\ER_0(f)$.

For the worst-case per-distribution bound, the lower bound follows from the expected-cost lower bound because worst-case query cost dominates expected query cost under $\mu^n$:
\begin{equation*}
\overline{D}^{\mu^n}_\varepsilon(f^n)\leq D^{\mu^n}_\varepsilon(f^n).
\end{equation*}
The upper bound is \cref{Lem_algo_CD}, \eqref{eq_Lem_algo_CD}. This proves
\begin{equation*}
\ED_0(f)-\varepsilon\ER_0(f)
\leq
\liminf_{n\to\infty}\frac{D^{\mu^n}_\varepsilon(f^n)}{n}
\leq
\limsup_{n\to\infty}\frac{D^{\mu^n}_\varepsilon(f^n)}{n}
\leq \ED_0(f).
\end{equation*}

Similarly, for the worst-case randomized bound, the lower bound follows from the randomized expected lower bound because
\begin{equation*}
\ER_\varepsilon(f^n)\leq R_\varepsilon(f^n),
\end{equation*}
and the upper bound follows from \cref{Lem_algo_CR}, \eqref{eq_Lem_algo_CR}. Hence
\begin{equation*}
(1-\varepsilon)\ER_0(f)
\leq
\liminf_{n\to\infty}\frac{R_\varepsilon(f^n)}{n}
\leq
\limsup_{n\to\infty}\frac{R_\varepsilon(f^n)}{n}
\leq \ER_0(f).
\end{equation*}
This completes the proof.
\end{proof}

We emphasize that the randomized expected quantity has a genuine limit: the matching lower and upper bounds imply convergence of $\ER_\varepsilon(f^n)/n$. For the per-distribution and worst-case variants, the theorem gives liminf/limsup bounds. Whether those auxiliary amortized limits always converge remains open in general.

% This file is intended to be included from Main.tex, for example via
% \input{OR_composition.tex}

\section{One-sided amortization for OR compositions}
\label{sec:or-composition}

In this section, we extend the amortized analysis to compositions of the form
$\operatorname{OR}_n\circ f$.  Let
$f:\mathcal{X}\longrightarrow\{0,1\}$,
$\mathcal{X}\subseteq\{0,1\}^{m}$,
be a Boolean partial function, and write
$\mathcal{X}_0:=f^{-1}(0)
\Mand
\mathcal{X}_1:=f^{-1}(1).$
The composed function is defined by
\[
(\operatorname{OR}_n\circ f)(x_1,\ldots,x_n)
:=\bigvee_{i=1}^{n}f(x_i).
\]
We consider algorithms with zero false-positive error and bounded
false-negative error.  More precisely, let $\mathfrak{A}_\varepsilon(f)$
denote the class of randomized query algorithms $A$ satisfying
\begin{align}
 \Pr[A(x)=1]&=0 \qquad \text{for every }x\in\mathcal{X}_0, \label{eq:one-sided-fp}\\
 \Pr[A(x)=0]&\leq \varepsilon \qquad \text{for every }x\in\mathcal{X}_1. \label{eq:one-sided-fn}
\end{align}
Thus, the error conditions are pointwise, rather than distributional.
Let $\mathcal{U}_0$ denote the set of probability distributions supported on
$\mathcal{X}_0$.

\begin{Def}[Zero-input expected query complexity]
\label{def:zero-input-query-complexity}
For $\varepsilon\in[0,1]$, define
\begin{align}
 \ER^{\mathcal{U}_0}_{\mathrm{FP}=0,\mathrm{FN}\leq\varepsilon}(f)
 &:={}
 \inf_{A\in\mathfrak{A}_\varepsilon(f)}
 \max_{\mu\in\mathcal{U}_0}
 \E_{x\sim\mu,R}[|A_R(x)|] \notag\\
 &=
 \inf_{A\in\mathfrak{A}_\varepsilon(f)}
 \max_{x\in\mathcal{X}_0}
 \E_R[|A_R(x)|].
 \label{eq:def-zero-input-query-complexity}
\end{align}
The equality follows because the expected query cost is linear in the input
distribution.
\end{Def}

For brevity, throughout this section we write
\[
 c_\varepsilon(f)
 :=\ER^{\mathcal{U}_0}_{\mathrm{FP}=0,\mathrm{FN}\leq\varepsilon}(f).
\]
We use
$\ER_{\mathrm{FP}=0,\mathrm{FN}\leq\varepsilon}(g)$ for the minimum,
over valid randomized algorithms for $g$, of the maximum expected number of
queries on an input.  Similarly,
$R_{\mathrm{FP}=0,\mathrm{FN}\leq\varepsilon}(g)$ denotes the usual
worst-case randomized query complexity under the same one-sided error
requirements.

The main result of this section is the following exact characterization.

\begin{Thm}[One-sided amortization for OR compositions]
\label{thm:or-composition}
For every Boolean partial function $f$ and every $\varepsilon\in[0,1]$,
\[
 \lim_{n\to\infty}
 \frac{
 \ER_{\mathrm{FP}=0,\mathrm{FN}\leq\varepsilon}
 (\operatorname{OR}_n\circ f)}{n}
 =c_\varepsilon(f).
\]
Moreover, for every $\varepsilon\in(0,1]$,
\[
 \lim_{n\to\infty}
 \frac{
 R_{\mathrm{FP}=0,\mathrm{FN}\leq\varepsilon}
 (\operatorname{OR}_n\circ f)}{n}
 =c_\varepsilon(f).
\]
\end{Thm}

The proof has two main ingredients.  The lower bound follows from a standard
coordinate embedding together with a minimax identity.  For the upper bound,
we first inspect a sublinear random set of coordinates exactly and then run a
single-copy algorithm on the remaining coordinates.

\subsection{A minimax identity}

The coordinate embedding below naturally constructs, for each
$\mu\in\mathcal{U}_0$, a single-copy algorithm that may depend on $\mu$.
We therefore record the relevant minimax identity.

\begin{Lem}[One-sided minimax identity]
\label{lem:one-sided-minimax}
For every $\varepsilon\in[0,1]$,
\[
 c_\varepsilon(f)
 =
 \max_{\mu\in\mathcal{U}_0}
 \inf_{A\in\mathfrak{A}_\varepsilon(f)}
 \E_{x\sim\mu,R}[|A_R(x)|].
\]
\end{Lem}

\begin{proof}
A deterministic decision tree may be assumed never to query the same input
bit twice.  Hence every relevant deterministic tree has depth at most $m$,
and there are only finitely many such trees.  A randomized query algorithm can
therefore be identified with a probability distribution $p$ over a finite set
of deterministic trees.

The constraints in \eqref{eq:one-sided-fp} and \eqref{eq:one-sided-fn} are
linear constraints on $p$.  Consequently, the set $\mathcal{P}_\varepsilon$
of admissible randomized algorithms is a compact convex polytope.  For a
deterministic tree $T$, let $q_T(x)$ be the number of queries made by $T$ on
$x$.  Then
\[
 c_\varepsilon(f)
 =
 \min_{p\in\mathcal{P}_\varepsilon}
 \max_{\mu\in\mathcal{U}_0}
 \sum_{T,x}p(T)\mu(x)q_T(x).
\]
The payoff is bilinear in $p$ and $\mu$.  The finite-dimensional minimax
theorem therefore gives
\[
 \min_{p\in\mathcal{P}_\varepsilon}
 \max_{\mu\in\mathcal{U}_0}
 \sum_{T,x}p(T)\mu(x)q_T(x)
 =
 \max_{\mu\in\mathcal{U}_0}
 \min_{p\in\mathcal{P}_\varepsilon}
 \sum_{T,x}p(T)\mu(x)q_T(x),
\]
which is the claimed identity.
\end{proof}

\subsection{The embedding lower bound}

The lower bound is exact for every number of coordinates.

\begin{Prop}[Embedding lower bound]
\label{prop:or-embedding-lower-bound}
For every $n\geq 1$ and every $\varepsilon\in[0,1]$,
\[
 \ER_{\mathrm{FP}=0,\mathrm{FN}\leq\varepsilon}
 (\operatorname{OR}_n\circ f)
 \geq n c_\varepsilon(f).
\]
Consequently,
\[
 R_{\mathrm{FP}=0,\mathrm{FN}\leq\varepsilon}
 (\operatorname{OR}_n\circ f)
 \geq n c_\varepsilon(f).
\]
\end{Prop}

\begin{proof}
Let $B$ be a randomized query algorithm for
$\operatorname{OR}_n\circ f$ with zero false-positive error and
false-negative error at most $\varepsilon$.  Fix a distribution
$\mu\in\mathcal{U}_0$.  We construct a single-copy algorithm $A_\mu$ for $f$.
On input $x\in\mathcal{X}$, the algorithm proceeds as follows.

\begin{enumerate}
 \item Choose $I$ uniformly from $[n]$.
 \item Independently sample $Z_j\sim\mu$ for every $j\neq I$.
 \item Simulate $B$ on
 \[
  Y^n=(Z_1,\ldots,Z_{I-1},x,Z_{I+1},\ldots,Z_n).
 \]
 Queries made by $B$ to a sampled block $Z_j$ are answered internally, while
 queries to the $I$th block are made to the actual input $x$.
 \item Output the answer produced by $B$.
\end{enumerate}

If $f(x)=0$, then every coordinate of $Y^n$ is a zero-input of $f$.
Therefore $(\operatorname{OR}_n\circ f)(Y^n)=0$, and the zero
false-positive property of $B$ implies that $A_\mu$ never outputs $1$.
If $f(x)=1$, then the $I$th coordinate is a one-input, so
$(\operatorname{OR}_n\circ f)(Y^n)=1$.  Hence $A_\mu$ outputs $0$ with
probability at most $\varepsilon$.  Thus
$A_\mu\in\mathfrak{A}_\varepsilon(f)$.

For an input $y^n$ and random string $r$, let $q_i(y^n,r)$ denote the number
of queries that $B_r$ makes to the $i$th block.  When $x\sim\mu$, the vector
$Y^n$ has distribution $\mu^n$.  It follows that
\begin{align*}
 \E_{x\sim\mu,R}[|A_{\mu,R}(x)|]
 &=
 \frac{1}{n}\sum_{i=1}^{n}
 \E_{Y^n\sim\mu^n,R}[q_i(Y^n,R)]\\
 &=
 \frac{1}{n}
 \E_{Y^n\sim\mu^n,R}[|B_R(Y^n)|]\\
 &\leq
 \frac{1}{n}\max_{y^n}
 \E_R[|B_R(y^n)|].
\end{align*}
Therefore,
\[
 \inf_{A\in\mathfrak{A}_\varepsilon(f)}
 \E_{x\sim\mu,R}[|A_R(x)|]
 \leq
 \frac{1}{n}\max_{y^n}\E_R[|B_R(y^n)|].
\]
Taking the maximum over $\mu\in\mathcal{U}_0$ and applying
\cref{lem:one-sided-minimax} gives
\[
 c_\varepsilon(f)
 \leq
 \frac{1}{n}\max_{y^n}\E_R[|B_R(y^n)|].
\]
Finally, taking the infimum over all valid algorithms $B$ yields
\[
 \ER_{\mathrm{FP}=0,\mathrm{FN}\leq\varepsilon}
 (\operatorname{OR}_n\circ f)
 \geq n c_\varepsilon(f).
\]
Since worst-case query cost dominates expected query cost, the second claim
follows as well.
\end{proof}

\subsection{The expected-cost upper bound}

We now prove a matching upper bound.  The random subsampling step ensures that
inputs containing many one-coordinates are detected with high probability,
while inputs containing few one-coordinates incur only a sublinear additional
cost.

\begin{Prop}[Expected-cost upper bound]
\label{prop:or-expected-upper-bound}
For every $\varepsilon\in[0,1]$,
\[
 \limsup_{n\to\infty}
 \frac{
 \ER_{\mathrm{FP}=0,\mathrm{FN}\leq\varepsilon}
 (\operatorname{OR}_n\circ f)}{n}
 \leq c_\varepsilon(f).
\]
\end{Prop}

\begin{proof}
Fix $\eta>0$, and choose
$A\in\mathfrak{A}_\varepsilon(f)$ satisfying
\[
 \max_{x\in\mathcal{X}_0}\E_R[|A_R(x)|]
 \leq c_\varepsilon(f)+\eta.
\]
We may assume that $A$ makes at most $m$ queries, since repeated queries can be
removed.

Set $s:=\lceil n^{2/3}\rceil$.  We construct an algorithm $B_n$ for
$\operatorname{OR}_n\circ f$ as follows.

\begin{enumerate}
 \item Choose a uniformly random subset $J\subseteq[n]$ of size $s$.
 \item Query all $m$ bits of every block indexed by $J$.  If one of these
 blocks is a one-input of $f$, output $1$ and halt.
 \item Otherwise, run independent copies of $A$ on every block indexed by
 $J^c$ and output the OR of their answers.
\end{enumerate}

If every block is a zero-input, then the exact inspection in the second step
finds no one-input and every copy of $A$ outputs $0$ with probability one.
Thus $B_n$ has zero false-positive error.

Suppose that the input contains $k\geq1$ one-input blocks.  If one of them lies
in $J$, then $B_n$ outputs $1$ exactly.  Otherwise, all $k$ one-input blocks
remain in $J^c$, and the copies of $A$ use independent randomness.  Therefore,
the probability that all of them output $0$ is at most
$\varepsilon^k\leq\varepsilon$.  Hence the false-negative error of $B_n$ is at
most $\varepsilon$.

It remains to bound the expected query cost.
Let \(p_k\) denote the probability that \(J\) contains no one-input
block. If \(s\leq n-k\), then
\begin{align}
 p_k
 &=\frac{\binom{n-k}{s}}{\binom{n}{s}}
   =\prod_{j=0}^{s-1}\left(1-\frac{k}{n-j}\right) \notag\\
 &\leq \left(1-\frac{k}{n}\right)^s
 \leq \exp\left(-\frac{sk}{n}\right).
 \label{eq:or-miss-probability}
\end{align}
If \(s>n-k\), then \(p_k=0\), and the same bound holds trivially.

The exact inspection uses at most \(ms\) queries. Conditioned on the
event that \(J\) contains no one-input block, the set \(J^c\) contains
\(k\) one-input blocks and \(n-s-k\) zero-input blocks. The expected
cost of the final step is therefore at most
\[
 (n-s-k)(c_\varepsilon(f)+\eta)+km
 \leq (n-k)(c_\varepsilon(f)+\eta)+km.
\]
Since the final step is performed only when \(J\) contains no
one-input block, it follows that
\begin{align*}
 \E_R[|B_{n,R}|]
 &\leq
 ms+p_k\bigl((n-k)(c_\varepsilon(f)+\eta)+km\bigr)\\
 &\leq
 n(c_\varepsilon(f)+\eta)+ms
 +mk\exp\left(-\frac{sk}{n}\right).
\end{align*}
Since $u e^{-u}\leq1$ for every $u\geq0$, we have
\[
 k\exp\left(-\frac{sk}{n}\right)\leq\frac{n}{s}.
\]
Thus, uniformly over all inputs,
\[
 \E_R[|B_{n,R}|]
 \leq
 n(c_\varepsilon(f)+\eta)
 +m\left(s+\frac{n}{s}\right)
 =n(c_\varepsilon(f)+\eta)+o(n).
\]
Dividing by $n$, taking the limit superior, and then letting $\eta\to0$ proves
the proposition.
\end{proof}

Combining \cref{prop:or-embedding-lower-bound,prop:or-expected-upper-bound}
proves the expected-cost statement in \cref{thm:or-composition}.

\subsection{The worst-case-cost upper bound}

To obtain a worst-case query bound, we truncate the preceding expected-cost
algorithm.  The truncation always outputs $0$, so it can introduce only
false-negative error.  We first record continuity from below in the allowable
false-negative probability.

\begin{Lem}[Continuity in the false-negative parameter]
\label{lem:or-continuity}
For every $\varepsilon\in(0,1]$,
\[
 \lim_{\varepsilon'\uparrow\varepsilon}c_{\varepsilon'}(f)
 =c_\varepsilon(f).
\]
\end{Lem}

\begin{proof}
Monotonicity gives
$c_{\varepsilon'}(f)\geq c_\varepsilon(f)$ whenever
$\varepsilon'<\varepsilon$.  Fix $\eta>0$, and choose
$A\in\mathfrak{A}_\varepsilon(f)$ satisfying
\[
 \max_{x\in\mathcal{X}_0}\E_R[|A_R(x)|]
 \leq c_\varepsilon(f)+\eta.
\]
Let $Z$ be the deterministic zero-error algorithm that queries all $m$ input
bits and evaluates $f$ exactly.  For $0<\varepsilon'<\varepsilon$, run $A$
with probability
\[
 \lambda:=\frac{\varepsilon'}{\varepsilon}
\]
and run $Z$ otherwise.  The resulting algorithm has zero false-positive error
and false-negative error at most
$\lambda\varepsilon=\varepsilon'$.  Its maximum expected cost on zero-inputs
is at most
\[
 \lambda(c_\varepsilon(f)+\eta)+(1-\lambda)m.
\]
Taking $\varepsilon'\uparrow\varepsilon$ and then $\eta\to0$ proves the claim.
The case $\varepsilon=1$ is included in the same argument.
\end{proof}

\begin{Prop}[Worst-case-cost upper bound]
\label{prop:or-worst-case-upper-bound}
For every $\varepsilon\in(0,1]$,
\[
 \limsup_{n\to\infty}
 \frac{
 R_{\mathrm{FP}=0,\mathrm{FN}\leq\varepsilon}
 (\operatorname{OR}_n\circ f)}{n}
 \leq c_\varepsilon(f).
\]
\end{Prop}

\begin{proof}
For $\varepsilon=1$, the algorithm that always outputs $0$ makes no queries,
and both sides of the claimed inequality are zero.  We may therefore assume
that $0<\varepsilon<1$.

\paragraph{Construction and query cost.}
Fix $0<\varepsilon'<\varepsilon$ and $\eta>0$.  Choose
$A\in\mathfrak{A}_{\varepsilon'}(f)$ such that
\[
 \max_{x\in\mathcal{X}_0}\E_R[|A_R(x)|]
 \leq c_{\varepsilon'}(f)+\eta.
\]
As before, we may assume that $A$ makes at most $m$ queries.  Apply the
subsampling algorithm from the proof of
\cref{prop:or-expected-upper-bound} with
\[
 s:=\lceil n^{2/3}\rceil.
\]
Set $K:=\lceil n^{2/3}\rceil$ and truncate the algorithm, outputting $0$,
just before its total number of queries would exceed
\[
 L_n
 :=n(c_{\varepsilon'}(f)+\eta)+ms+mK+n^{2/3}.
\]
Thus its worst-case query cost is at most $L_n$, where
\[
 L_n=n(c_{\varepsilon'}(f)+\eta)+O(n^{2/3}).
\]

\paragraph{Truncation probability.}
We show that the probability of truncation is $o(1)$ uniformly over all
inputs.  Let $k$ be the number of one-input blocks.

First suppose that $k>K$.  The algorithm reaches the final execution stage
only if the sampled set misses every one-input block.  By
\eqref{eq:or-miss-probability}, the probability of this event is at most
\[
 \exp\left(-\frac{sk}{n}\right)
 \leq \exp(-\Omega(n^{1/3})).
\]
Thus the truncation probability is exponentially small in this case.

Now suppose that $k\leq K$.  Fix a sampled set $J$ such that
$f(x_i)=0$ for every $i\in J$.  For each $i\in J^c$, let $X_i$ denote the
number of queries made by $A$ on input $x_i$, and set
\[
 S_J:=\sum_{i\in J^c}X_i.
\]
Conditional on $J$, the random variables $X_i$ are independent and satisfy
$0\leq X_i\leq m$.  Since $J^c$ contains $k$ one-input blocks,
\[
 \E[S_J\mid J]
 \leq (n-s-k)(c_{\varepsilon'}(f)+\eta)+km
 \leq n(c_{\varepsilon'}(f)+\eta)+mK.
\]
The exact inspection of the sampled blocks uses at most $ms$ queries.
Consequently, truncation can occur only if
\[
 S_J>L_n-ms
 =n(c_{\varepsilon'}(f)+\eta)+mK+n^{2/3},
\]
which implies
\[
 S_J-\E[S_J\mid J]>n^{2/3}.
\]

We use Hoeffding's inequality in the following form: if
$X_1,\ldots,X_N$ are independent and $a_i\leq X_i\leq b_i$, then
\[
 \Pr\left[
 \sum_{i=1}^N X_i-\E\left[\sum_{i=1}^N X_i\right]\geq t
 \right]
 \leq
 \exp\left(-\frac{2t^2}{\sum_{i=1}^N(b_i-a_i)^2}\right).
\]
Conditional on $J$, apply this inequality with
$N=n-s$, $a_i=0$, $b_i=m$, and $t=n^{2/3}$.  We obtain
\[
 \Pr\left[
 S_J-\E[S_J\mid J]>n^{2/3}\,\middle|\,J
 \right]
 \leq
 \exp\left(-\frac{2n^{4/3}}{(n-s)m^2}\right)
 \leq
 \exp\left(-\frac{2n^{1/3}}{m^2}\right)
 =o(1).
\]
This bound holds for every $J$ containing no one-input block.  If $J$
contains a one-input block, the algorithm halts during the exact inspection
and cannot be truncated.  Hence the overall truncation probability is
$o(1)$ uniformly over all inputs.

\paragraph{Error probability.}
Before truncation, the algorithm has zero false-positive error and
false-negative error at most $\varepsilon'$.  Since truncation outputs $0$,
it cannot create a false positive and increases the false-negative
probability by at most the truncation probability.  The false-negative error
of the truncated algorithm is therefore at most
\[
 \varepsilon'+o(1).
\]
Since $\varepsilon'<\varepsilon$, this is at most $\varepsilon$ for all
sufficiently large $n$.

\paragraph{Conclusion.}
The truncated algorithm is therefore valid for all sufficiently large $n$
and has worst-case query cost at most $L_n$.  It follows that
\[
 \limsup_{n\to\infty}
 \frac{
 R_{\mathrm{FP}=0,\mathrm{FN}\leq\varepsilon}
 (\operatorname{OR}_n\circ f)}{n}
 \leq c_{\varepsilon'}(f)+\eta.
\]
Letting $\eta\to0$ and then $\varepsilon'\uparrow\varepsilon$, and applying
\cref{lem:or-continuity}, proves the proposition.
\end{proof}

The worst-case statement of \cref{thm:or-composition} follows by combining
\cref{prop:or-embedding-lower-bound,prop:or-worst-case-upper-bound}.

\begin{Rem}
The expected-cost identity remains valid at $\varepsilon=0$.  The worst-case
argument above requires $\varepsilon>0$, because truncation may introduce a
small false-negative probability.  Determining whether the corresponding
zero-error worst-case limit always equals $c_0(f)$ is a separate question.
\end{Rem}

\section*{Acknowledgement}
The author would like to thank Ashwin Nayak and Dave Touchette for their kindness and support, Atsuya Hasegawa and Masayuki Miyamoto for monthly discussions,
and the author's former adviser Michał Oszmaniec for his kindness and allowing the author to take time for the author's individual research at CFT PAN.

The author acknowledges the support of the Natural Sciences and Engineering Research Council of Canada (NSERC) [funding reference no. ALLRP-578455-2022]. This research was supported in part by the Institute for Quantum Computing.
\section*{Disclosure of AI Use}
The author used OpenAI's ChatGPT (GPT-5.5 and GPT-5.6) solely to improve the readability of the manuscript and to identify minor typographical and grammatical errors. All AI-generated suggestions were reviewed by the author, who takes full responsibility for the final content of the manuscript.

\bibliography{/Users/daikisuruga/Dropbox/Citations/mathematics_D, /Users/daikisuruga/Dropbox/Citations/computational_D, /Users/daikisuruga/Dropbox/Citations/query_D, /Users/daikisuruga/Dropbox/Citations/quant_info_D, /Users/daikisuruga/Dropbox/Citations/comm_comp_D, /Users/daikisuruga//Dropbox/Citations/books_D}

\newcommand{\etalchar}[1]{$^{#1}$}
\begin{thebibliography}{BDGKW20}

\bibitem[ABB{\etalchar{+}}16]{ABB+16}
Andris Ambainis, Kaspars Balodis, Aleksandrs Belovs, Troy Lee, Miklos Santha,
  and Juris Smotrovs.
\newblock Separations in query complexity based on pointer functions.
\newblock In {\em 48th annual Symposium on Theory of Computing}, pages
  800--813, 2016.

\bibitem[ACLGT10]{ACLT10}
Andris Ambainis, Andrew~M Childs, Fran{\c{c}}ois Le~Gall, and Seiichiro Tani.
\newblock The quantum query complexity of certification.
\newblock {\em Quantum Information \& Computation}, 10(3):181--189, 2010.

\bibitem[AMRR11]{AMRR11}
Andris Ambainis, Lo{\"\i}ck Magnin, Martin Roetteler, and J{\'e}r{\'e}mie
  Roland.
\newblock Symmetry-assisted adversaries for quantum state generation.
\newblock In {\em 26th annual Conference on Computational Complexity}, pages
  167--177. IEEE, 2011.

\bibitem[BB19]{BB19}
Eric Blais and Joshua Brody.
\newblock Optimal separation and strong direct sum for randomized query
  complexity.
\newblock In {\em 34th Computational Complexity Conference}, pages 1--17, 2019.

\bibitem[BBCR09]{BBCR09}
Boaz Barak, Mark Braverman, Xi~Chen, and Anup Rao.
\newblock Direct sums in randomized communication complexity.
\newblock {\em Electronic Colloquium on Computational Complexity}, 44, 2009.

\bibitem[BDB20]{BB20}
Shalev Ben-David and Eric Blais.
\newblock A tight composition theorem for the randomized query complexity of
  partial functions.
\newblock In {\em 61st annual Symposium on Foundations of Computer Science},
  pages 240--246, 2020.

\bibitem[BDB25]{BB25}
Shalev Ben-David and Eric Blais.
\newblock Direct product theorems for randomized query complexity.
\newblock In {\em 66th Annual Symposium on Foundations of Computer Science},
  pages 710--733, 2025.

\bibitem[BDGKW20]{BGK+20}
Shalev Ben-David, Mika G\"{o}\"{o}s, Robin Kothari, and Thomas Watson.
\newblock {When Is Amplification Necessary for Composition in Randomized Query
  Complexity?}
\newblock In {\em Approximation, Randomization, and Combinatorial Optimization.
  Algorithms and Techniques (APPROX/RANDOM 2020)}, volume 176, pages
  28:1--28:16, 2020.

\bibitem[BDK18]{BK18}
Shalev Ben-David and Robin Kothari.
\newblock Randomized query complexity of sabotaged and composed functions.
\newblock {\em Theory of Computing}, 14(1):1--27, 2018.

\bibitem[BdW02]{BdW02}
Harry Buhrman and Ronald de~Wolf.
\newblock Complexity measures and decision tree complexity: a survey.
\newblock {\em Theoretical Computer Science}, 288(1):21--43, 2002.

\bibitem[BI87]{BI87}
Manuel Blum and Russell Impagliazzo.
\newblock Generic oracles and oracle classes.
\newblock In {\em 28th annual Symposium on Foundations of Computer Science},
  pages 118--126, 1987.

\bibitem[BKLS23]{BTLS23}
Joshua Brody, Jae~Tak Kim, Peem Lerdputtipongporn, and Hariharan Srinivasulu.
\newblock A strong {XOR} lemma for randomized query complexity.
\newblock {\em Theory of Computing}, 19(1):1--14, 2023.

\bibitem[BKST24]{BKST24}
Guy Blanc, Caleb Koch, Carmen Strassle, and Li-Yang Tan.
\newblock {A Strong Direct Sum Theorem for Distributional Query Complexity}.
\newblock In {\em 39th Computational Complexity Conference}, volume 300, pages
  16:1--16:30, 2024.

\bibitem[BR11]{BR11}
Mark Braverman and Anup Rao.
\newblock Information equals amortized communication.
\newblock In {\em 52nd annual Symposium on Foundations of Computer Science},
  pages 748--757, 2011.

\bibitem[Bra17]{Bra17}
Mark Braverman.
\newblock Interactive information complexity.
\newblock {\em SIAM Review}, 59(4):803--846, 2017.

\bibitem[BYJKS04]{BJKS04}
Ziv Bar-Yossef, T.S. Jayram, Ravi Kumar, and D.~Sivakumar.
\newblock An information statistics approach to data stream and communication
  complexity.
\newblock {\em Journal of Computer and System Sciences}, 68(4):702--732, 2004.

\bibitem[CSWY01]{CSWY01}
Amit Chakrabarti, Yaoyun Shi, Anthony Wirth, and Andrew Yao.
\newblock Informational complexity and the direct sum problem for simultaneous
  message complexity.
\newblock In {\em 42nd annual Symposium on Foundations of Computer Science},
  pages 270--278, 2001.

\bibitem[Dru11]{Dru11}
Andrew Drucker.
\newblock Improved direct product theorems for randomized query complexity.
\newblock In {\em 26th Computational Complexity Conference}, pages 1--11, 2011.

\bibitem[FKNN95]{FKNN95}
Tom{\'a}s Feder, Eyal Kushilevitz, Moni Naor, and Noam Nisan.
\newblock Amortized communication complexity.
\newblock {\em SIAM Journal on computing}, 24(4):736--750, 1995.

\bibitem[GJPW17]{GJPW17}
Mika G\"{o}\"{o}s, T.~S. Jayram, Toniann Pitassi, and Thomas Watson.
\newblock {Randomized Communication vs. Partition Number}.
\newblock In {\em 44th International Colloquium on Automata, Languages, and
  Programming (ICALP 2017)}, volume~80, pages 52:1--52:15, 2017.

\bibitem[GM21]{GM21}
Mika G{\"o}{\"o}s and Gilbert Maystre.
\newblock A majority lemma for randomised query complexity.
\newblock In {\em 36th Computational Complexity Conference}, 2021.

\bibitem[GNW11]{GNW11}
Oded Goldreich, Noam Nisan, and Avi Wigderson.
\newblock On {Yao}'s {XOR}-lemma.
\newblock {\em Studies in Complexity and Cryptography}, pages 273--301, 2011.

\bibitem[HH86]{HH86}
Juris Hartmanis and Lane~A. Hemachandra.
\newblock One-way functions, robustness, and the non-isomorphism of
  {NP}-complete sets.
\newblock Technical Report TR86-796, Cornell University, 1986.

\bibitem[IJKW08]{IJKW08}
Russell Impagliazzo, Ragesh Jaiswal, Valentine Kabanets, and Avi Wigderson.
\newblock Uniform direct product theorems: simplified, optimized, and
  derandomized.
\newblock In {\em 40th annual Symposium on Theory of Computing}, pages
  579--588, 2008.

\bibitem[IW97]{IW97}
Russell Impagliazzo and Avi Wigderson.
\newblock {P= BPP} if {E} requires exponential circuits: Derandomizing the xor
  lemma.
\newblock In {\em 29th annual Symposium on Theory of Computing}, pages
  220--229, 1997.

\bibitem[Jai20]{Jain20}
Rahul Jain.
\newblock A near-optimal direct-sum theorem for communication complexity.
\newblock {\em arXiv preprint arXiv:2008.07188}, 2020.

\bibitem[JK09]{JK09}
Rahul Jain and Hartmut Klauck.
\newblock New results in the simultaneous message passing model via information
  theoretic techniques.
\newblock In {\em 24th Computational Complexity Conference}, pages 369--378,
  2009.

\bibitem[JK22]{JK22}
Rahul Jain and Srijita Kundu.
\newblock A direct product theorem for quantum communication complexity with
  applications to device-independent {QKD}.
\newblock In {\em 62nd annual Symposium on Foundations of Computer Science},
  pages 1285--1295, 2022.

\bibitem[JKS10]{JKS10}
Rahul Jain, Hartmut Klauck, and Miklos Santha.
\newblock Optimal direct sum results for deterministic and randomized decision
  tree complexity.
\newblock {\em Information Processing Letters}, 110(20):893--897, 2010.

\bibitem[JPY12]{JPY12}
Rahul Jain, Attila Pereszl{\'e}nyi, and Penghui Yao.
\newblock A direct product theorem for the two-party bounded-round public-coin
  communication complexity.
\newblock In {\em 53rd annual Symposium on Foundations of Computer Science},
  pages 167--176, 2012.

\bibitem[JRS03]{JRS03d}
Rahul Jain, Jaikumar Radhakrishnan, and Pranab Sen.
\newblock A direct sum theorem in communication complexity via message
  compression.
\newblock In {\em 30th International Colloquium on Automata, Languages and
  Programming}, pages 300--315, 2003.

\bibitem[KMSY14]{KMSY14}
Gillat Kol, Shay Moran, Amir Shpilka, and Amir Yehudayoff.
\newblock Direct sum fails for zero error average communication.
\newblock In {\em 5th conference on Innovations in Theoretical Computer
  Science}, pages 517--522, 2014.

\bibitem[Mon13]{Mon13}
Ashley Montanaro.
\newblock A composition theorem for decision tree complexity.
\newblock {\em arXiv preprint arXiv:1302.4207}, 2013.

\bibitem[MS15]{MS15}
Sagnik Mukhopadhyay and Swagato Sanyal.
\newblock Towards better separation between deterministic and randomized query
  complexity.
\newblock In {\em 35th IARCS Annual Conference on Foundations of Software
  Technology and Theoretical Computer Science}, page 206, 2015.

\bibitem[MWY13]{MWY13}
Marco Molinaro, David~P Woodruff, and Grigory Yaroslavtsev.
\newblock Beating the direct sum theorem in communication complexity with
  implications for sketching.
\newblock In {\em 24th annual ACM-SIAM symposium on Discrete algorithms}, pages
  1738--1756, 2013.

\bibitem[Sha01]{Sha01}
Ronen Shaltiel.
\newblock Towards proving strong direct product theorems.
\newblock In {\em 16th Computational Complexity Conference}, pages 107--117,
  2001.

\bibitem[Tar89]{Tar89}
G{\'a}bor Tardos.
\newblock Query complexity, or why is it difficult to separate {NP}$^{A}$
  $\cap$ co{NP}$^{A}$ from {P}$^{A}$ by random oracles {A}?
\newblock {\em Combinatorica}, 9(4):385--392, December 1989.

\end{thebibliography}
%\bibliography{/Users/Suruga_research/Dropbox/Citations/Scheduling, /Users/Suruga_research/Dropbox/Citations/comm_comp_D, /Users/Suruga_research/Dropbox/Citations/books_D}
\bibliographystyle{alpha}
\appendix
\section{Omitted proofs}\label{Appendix_Proof}

\begin{Fact1}
There is a function $f_n:\Bset^n \to \Bset$, a distribution $\mu_n$ on $\Bset^n$ and an error $0 < \varepsilon_n <1$ such that 
\begin{equation*}
\overline{D}^{\mu_n}_0(f_n)
= O(1) \Mand D^{\mu_n}_\varepsilon(f_n) = \Omega(n).
\end{equation*}
\end{Fact1}
\begin{proof}
Define first the distribution on $\Bset^n$ as $P(x_1 = 1) = 1/n$, and other $x_i$'s are uniformly distributed.
And define $f(x_1, \ldots, x_n) = \bigoplus_{2 \leq i \leq n} x_i$ if $x_1 = 1$ and $f(x_1, \ldots, x_n) = 0$ if $x_1 = 0$.
\par
The function and distribution defined in this manner satisfy $\ED_0(f) = O(1)$ because an algorithm that checks $x_1$ deterministically and checks all the other bits only if $x_1 = 1$ has the expectation cost equal to $2$, because of the definition of the distribution.
\par
To prove $D^{\mu}_\varepsilon(f) = \Omega(n)$, set $\varepsilon := 1/n^2$. We first observe that any query algorithm $\calA \in [f, \mu, \varepsilon]$ becomes that of the $(n-1)$-bit parity function under the uniform distribution, by internally setting $x_1 = 1$ and running $\calA$. This modification changes the error $\varepsilon = 1/n^2$ to $1/n$. (This can be checked directly from the definition of distribution.) As we know $D^\mathsf{Uni}_\varepsilon(\mathsf{Parity}) = \Omega(n)$, $D^\mu_\varepsilon(f) = \Omega(n)$ holds.
\end{proof}

\begin{Fact}\label{Fact_appendix_cont}
There exists a total relation $f \subseteq B \times [m]~(B \subseteq \Bset^m)$ and  an error $\varepsilon$ such that
\begin{equation*}
\ER_0(f) = \Theta(1) \Mand R_\varepsilon(f) = \Omega(m).
\end{equation*}
\end{Fact}
\begin{proof}
Define the input space $B \subseteq \Bset^m$ to be a set of bit strings that have at least $m/2$ fraction of $1$'s, and $f(x) \subseteq [m]$ for $x \in B$ is a set of coordinates $i$ satisfying $x_i =1$.
Then $\ER_0(f) = \Theta(1)$: repeatedly sample a uniformly random coordinate $i$, query $x_i$, and output $i$ once $x_i=1$. Since at least half of the coordinates are $1$, the expected number of queries is at most $2$ and the algorithm has zero error.
\par
On the other hand, any deterministic algorithm must have the complexity $D(f) = \Omega(m)$. This yields $R_\varepsilon(f) = \Omega(m)$ for a smaller $\varepsilon$ due to~\cref{Fact_appendix_cont}.
\end{proof}

\begin{Fact2}
For any total relation $f \subseteq \mathcal{X} \times \mathcal{Y}$, let $\varepsilon < 1/|\mathcal{X}|$. Then $R_\varepsilon(f) = D(f)$.
\end{Fact2}

\begin{proof}
Let $\pi \in [f, \varepsilon]$ be an optimal algorithm: $\Cost(\pi) = R_\varepsilon(f)$ and denote the set of randomness used in $\pi$ by $R$.
The set of randomness that may make a mistake on some input $x \in \mathcal{X}$ is then defined as
\begin{equation*}
R_\mathsf{wrong} := \{r \in R \mid \exists x \textrm{~s.t.~} \pi_r(x) \notin f(x)\}
\end{equation*}
where $\pi_r(x)$ denotes the output of the algorithm $\pi$ when the input is $x \in \mathcal{X}$ and the randomness is $r \in R$.
\par
Since $\pi$ has the worst-case error $\leq \varepsilon$, for any $x \in \mathcal{X}$, 
\begin{equation*}
\Pr_R(\{r \in R \mid \pi_r(x) \notin f(x)\}) < \frac{1}{|\mathcal{X}|}.
\end{equation*}
By summing up over all $x \in \mathcal{X}$, this leads to 
\begin{equation*}
\Pr_R(R_\mathsf{wrong}) \leq \sum_{x \in \mathcal{X}} \Pr_R(\{r \in R \mid \pi_r(x) \notin f(x)\}) < 1.
\end{equation*}
This means that there exists $r_\mathsf{good} \in R\setminus R_\mathsf{wrong}$, which satisfies
$\pi_{r_\mathsf{good}}(x) \in f(x)$ for any $x \in \mathcal{X}$.
Fixing the randomness to $r_\mathsf{good}$, the deterministic algorithm $\pi_{r_\mathsf{good}}$ always outputs a correct value.
This means $D(f) \leq \Cost(\pi_{r_\mathsf{good}})$.
Since the new algorithm $\pi_{r_\mathsf{good}}$ must have a smaller complexity than (or at most equal to) $\Cost(\pi)$, 
we also observe $\Cost(\pi_{r_\mathsf{good}}) \leq R_\varepsilon(f)$.
Together with the trivial relation $R_\varepsilon(f) \leq D(f)$, these arguments show the desired statement.
\end{proof}

\end{document}